%% file: main.tex
\documentclass[11pt]{article}

\usepackage[letterpaper,top=1in,left=1in,right=1in,bottom=1in]{geometry}
\usepackage{bm}

\input{common}
\input{macros}

\input{localdef_yu}
\input{localdef_hoyee}
\input{localdef_shaddin}

\input{localdef_ehsan}



\begin{document}

\TitleAndAuthors

%%%%%%%%%%%%%%%%
%% two options for title page numbering:
%% with title page not having page number:
\Eat{
\begin{titlepage}
\maketitle
\thispagestyle{empty}
\begin{abstract}
\input{abstract}
\end{abstract}
\end{titlepage}
}
%%%%%%%%%%%%%%%%
%% with title page having page number:
\maketitle
\begin{abstract}
\input{abstract}
\end{abstract}
\pagebreak
%%%%%%%%%%%%%%%%

\input{introduction}

\input{framework}

\input{lottery}

\input{signaling.tex}

\input{auction}

\input{voting}

\input{game}

\input{hardness}

\input{fourier}

%%%%%%%%%%%%%%%%%%%%%%%%%%%%%%%%%%%%%%%%%%%%%%%%%%

\bibliographystyle{alpha}
\bibliography{bib/agt,bib/signal,bib/general}

\end{document}

%% file: common.tex
\usepackage{amssymb,amsmath}
\usepackage{amsthm}

\usepackage{fullpage}
\usepackage{hyperref}
\usepackage{xspace}

\newtheorem{theorem}{Theorem}[section]

\newtheorem{lemma}[theorem]{Lemma}

\newtheorem{definition}[theorem]{Definition}

\newtheorem{observation}[theorem]{Observation}

\newtheorem{prop}[theorem]{Proposition}

\def\QED{{\phantom{x}} \hfill \ensuremath{\rule{1.3ex}{1.3ex}}}

\numberwithin{equation}{section}

\makeatletter
\renewcommand \thesection {\@arabic\c@section}      
\renewcommand \thetable {\@arabic\c@table}
\renewcommand \thefigure {\@arabic\c@figure}
\makeatother

\def \TitleAndAuthors{
\title{Mixture Selection, Mechanism Design, and Signaling \thanks{This work was supported by NSF grants CCF-1350900, CCF-1423618, CCF-0964481, CCF-1111270, and the last author's Simons Investigator Award from the Simons Foundation.}}
\author{
Yu Cheng
\and
Ho Yee Cheung
\and
Shaddin Dughmi
\and
Ehsan Emamjomeh-Zadeh
\and 
Li Han
\and
Shang-Hua Teng
\thanks{ University of Southern California,
\{yu.cheng.1, hoyeeche, shaddin, emamjome, li.han, shanghua\}@usc.edu.}
}}

%% file: macros.tex
\newcommand{\cross}{\times}

\newcommand{\set}[1]{\left\{ #1 \right\}}

\newcommand{\floor}[1]{\lfloor {#1} \rfloor}
\newcommand{\ceil}[1]{\lceil {#1} \rceil}

\renewcommand{\tilde}{\widetilde}
\renewcommand{\bar}{\overline}

\DeclareMathOperator{\poly}{poly}

\def\ex{\qopname\relax n{E}}
\def\min{\qopname\relax n{min}}
\def\max{\qopname\relax n{max}}
\def\maxtwo{\qopname\relax n{max2}}
\def\Pr{\Prob}
\def\Ex{\Exp}
\newcommand{\RR}{\mathbb{R}}

\def\A{\mathcal{A}}
\def\B{\mathcal{B}}

\def\D{\mathcal{D}}
\def\E{\mathcal{E}}
\def\F{\mathcal{F}}
\def\G{\mathcal{G}}

\def\I{\mathcal{I}}

\def\T{\mathcal{T}}
\def\V{\mathcal{V}}

\def\eps{\epsilon}
\def\sse{\subseteq}

\newcommand{\eat}[1]{}

\newcommand{\maxi}[1]{\mbox{maximize} & {#1 } & \\}
\newcommand{\st}{\mbox{subject to} }
\newcommand{\con}[1]{&#1 & \\}

\newenvironment{lp}{\begin{equation}  \begin{array}{lll}}{\end{array}\end{equation}}
\newenvironment{lp*}{\begin{equation*}  \begin{array}{lll}}{\end{array}\end{equation*}}

%% file: localdef_yu.tex
\newtheorem{assumption}[theorem]{Assumption}

\def\Half{\frac{1}{2}}
\def\gclique{\ensuremath{g^{\text{(clique)}}}}
\def\T{\intercal}

%% file: localdef_hoyee.tex
\newcommand{\numact}{n}
\newcommand{\numppl}{k}
\newcommand{\qedsameline}{\belowdisplayskip=-11pt}
\newcommand{\gRev}{g^{\text{(rev)}}}

%% file: localdef_shaddin.tex
\usepackage{stackrel}
\newcommand{\corr}[1]{\stackrel[#1]{}{\approx}}
\DeclareMathOperator{\midd}{mid}
\newcommand{\gmid}{g^{(\midd)}}

\newcommand{\prob}{mixture selection }
\newcommand{\probnospace}{mixture selection}

\def\gVotesum{\ensuremath{g^{\text{(vote-sum)}}}}
\def\gVotethresh{\ensuremath{g^{\text{(vote-thresh)}}}}
\def\gVotesmooth{\ensuremath{g^{\text{(vote-smooth-thresh)}}}}

\def\FVotesum{\ensuremath{F^{\text{(vote-sum)}}}}
\def\FVotethresh{\ensuremath{F^{\text{(vote-thresh)}}}}

\def\OVotesum{\ensuremath{OPT^{\text{(vote-sum)}}}}
\def\OVotethresh{\ensuremath{OPT^{\text{(vote-thresh)}}}}

%% file: localdef_ehsan.tex
\newcommand{\BBB}{\ensuremath{\mathcal{B}}\xspace}
\newcommand{\CCC}{\ensuremath{\mathcal{C}}\xspace}
\newcommand{\DDD}{\ensuremath{\mathcal{D}}\xspace}
\newcommand{\EEE}{\ensuremath{\mathcal{E}}\xspace}
\newcommand{\FFF}{\ensuremath{\mathcal{F}}\xspace}
\newcommand{\GGG}{\ensuremath{\mathcal{G}}\xspace}

\newcommand{\III}{\ensuremath{\mathcal{I}}\xspace}

\DeclareMathOperator*{\Exp}{E}
\DeclareMathOperator*{\Prob}{Pr}

\newcommand{\Eat}[1]{}
\newcommand{\Norm}[2]{\ensuremath{{||#1||}_{#2}}\xspace}
\newcommand{\infNorm}[1]{\Norm{#1}{\infty}}
\newcommand{\ABS}[1]{\ensuremath{|#1|}\xspace}
\newcommand{\Ordinal}[1]{\ensuremath{{#1}^{\rm th}}}
\newcommand{\Set}[1]{\ensuremath{\{ #1 \}}}
\newcommand{\SpecSet}[2]{\ensuremath{\Set{#1 : #2}}}
\newcommand{\Ceil}[1]{\ensuremath{\left\lceil{#1}\right\rceil}}

\newcommand{\Size}[1]{\ensuremath{| #1 |}}

\def\numrow{n}
\def\numcol{m}

\def\ginp{\ensuremath{t}}

\def\gdom{\ensuremath{[-1, 1]^\numrow}\xspace}

\def\gindices{\ensuremath{[\numrow]}\xspace}

\def\lightSym{\alpha}
\newcommand{\Light}[1]%
{\ensuremath{#1\text{-Light}}\xspace}
\def\aLight{\Light{\lightSym}}

\def\stableSym{\beta}
\newcommand{\gStable}[1]%
{\ensuremath{#1\text{-stable}}\xspace}
\newcommand{\gstable}[1]%
{\ensuremath{#1\text{-stable}}\xspace}
\def\abStable{\gstable{\stableSym}}

\def\lipSym{c}
\newcommand{\Lip}[1]%
{\ensuremath{#1\text{-Lipschitz}}\xspace}
\def\cLip{\Lip{\lipSym}}

\newcommand{\Rel}[1]%
{\ensuremath{#1\text{-relaxation}}\xspace}

\def\xsparse{\ensuremath{\tilde{x}}\xspace}

\def\gSlope{\ensuremath{g^{\text{(slope)}}}}

\newcommand{\gLot}[1]{g^{\text{(lottery)}}_{#1}}

\def\Rev{\text{Rev}\xspace}
\def\Weight{w}

%% file: abstract.tex
We pose and study a fundamental algorithmic problem
which we term \emph{\probnospace},
arising as a building block in a number of game-theoretic applications:
Given a function $g$ from the $\numrow$-dimensional hypercube
to the bounded interval $[-1,1]$,
and an $\numrow \times \numcol$ matrix $A$
with bounded entries, maximize $g(Ax)$ over $x$
in the $\numcol$-dimensional simplex.
This problem arises naturally when one seeks to design a lottery over items
for sale in an auction,
or craft the posterior beliefs for agents in a Bayesian game
through the provision of information (a.k.a. signaling).

We present an approximation algorithm for this problem
when $g$ simultaneously satisfies two ``smoothness'' properties:
Lipschitz continuity with respect to the $L^\infty$ norm,
and \emph{noise stability}.
The latter notion, which we define and cater to our setting,
controls the degree to which low-probability --- and possibly correlated --- errors
in the inputs of $g$ can impact its output.
The approximation guarantee of our algorithm degrades gracefully
as a function of the Lipschitz continuity and noise stability of $g$.
In particular, when $g$ is both $O(1)$-Lipschitz continuous and $O(1)$-stable,
we obtain an (additive) polynomial-time approximation scheme (PTAS) for \probnospace.
We also show that neither assumption suffices by itself for an additive PTAS, and both assumptions together do not suffice for an additive fully polynomial-time approximation scheme (FPTAS).

We apply our algorithm for mixture selection
to a number of different game-theoretic applications,
focusing on problems from mechanism design and optimal signaling.
In particular, we make progress on a number of open problems
suggested in prior work by easily reducing them to \probnospace:
we resolve an important special case of the \emph{small-menu lottery design} problem
posed by Dughmi, Han, and Nisan~\cite{DHN14};
we resolve the problem of
\emph{revenue-maximizing signaling in Bayesian second-price auctions}
posed by Emek et al.~\cite{emeksignaling}
and Miltersen and Sheffet~\cite{miltersensignals};
we design a quasipolynomial-time approximation scheme
for the \emph{optimal signaling problem in normal form games}
suggested by Dughmi~\cite{D14};
and we design an approximation algorithm for the optimal signaling problem
in the voting model of Alonso and C\^{a}mara~\cite{Alonso14}.

%%% Local Variables: 
%%% mode: latex
%%% TeX-master: "main"
%%% End: 

%% file: introduction.tex
%% Introduction

\section{Introduction}\label{sec:introduction}

Lotteries, beliefs, mixed strategies ---
all are distributions arising as important objects in game theory.
It is unsurprising, therefore, that algorithmic game theory is rife
with algorithmic problems which --- implicitly or explicitly ---
optimize over the space of distributions,
or equivalently the simplex.
In this paper, we identify a family of algorithmic problems over the simplex
which arise over and over in game theory.
We term problems in this class \emph{\probnospace},
and examine their computational complexity. 

Given a function $g$ from the solid $\numrow$-dimensional hypercube
%(i.e., $[-1,1]^n$ or $[0,1]^n$)
to the bounded interval $[-1,1]$, and an integer $\numcol$,
we define the \emph{$\numcol$-dimensional \prob problem} for $g$ as follows.
The input to this problem is an  $\numrow \times \numcol$ matrix $A$ with bounded entries,
and the objective is to compute $x \in \Delta_\numcol$ maximizing $g(Ax)$. 
It is natural to expect that the computational complexity of \prob
depends crucially on the ``complexity'' of the function $g$.
We therefore identify two ``smoothness'' parameters of the function $g$
which control the extent to which \prob is tractable,
and derive a simple approximation algorithm with  guarantees degrading gracefully
in those parameters. Moreover, we present evidence --- in the form of hardness results ---
that smoothness in both senses is necessary for the kind of general results we obtain.

The first smoothness quantity is a familiar one,
namely Lipschitz continuity in the $L^\infty$ metric.
The second quantity, which we define and term \emph{noise stability},
borrows ideas from related definitions of stability in other contexts
(e.g. \cite{kkl88, moo10}), though is importantly different.
Informally, a function $g$ from the solid $\numrow$-dimensional hypercube
to the real numbers is \emph{$\beta$-noise stable}
(or $\beta$-stable for short)
if the random corruption of an $\alpha$-fraction of the $\numrow$ inputs to $g$,
with no individual input disproportionately likely to be corrupted,
does not decrease the output of $g$ by more than $\alpha \beta$. 
We note that a Fourier-analytic notion of stability is closely-related to ours ---
we elaborate on this connection in Section~\ref{sec:fourier}.

This paper lays out a framework for tackling \prob problems,
and presents a number of applications in mechanism design and optimal signaling in games.
Notably, we find that we resolve or make progress on a number of known open problems,
and some new ones, using our framework.

\subsection*{Our Results}

Our results for \prob can be viewed
as generalizing the main insights of Lipton et al.~\cite{lmm03}.
First, we show that when $g$ is noise stable and Lipschitz continuous,
and $x \in \Delta_\numcol$ is arbitrary, there is a sparse vector $\xsparse$
for which $g(A \xsparse)$ is not much smaller than  $g(A x)$.
The proof of this fact proceeds by sampling from $x$
and letting $\xsparse$ be the empirical distribution, as in \cite{lmm03}.
However, when $g$ is sufficiently noise stable and Lipschitz continuous,
we obtain a better tradeoff between the number of samples required
and the error introduced into the objective than does \cite{lmm03},
and this is crucial for our applications.
Our analysis bounds the expected difference between $g(A x)$  and $g(A \xsparse)$
as the sum of two terms:
The first term represents the error in the output of $g$
caused by the low-probability ``large errors'' in its $\numrow$ inputs,
and the second term represents the error in the output of $g$
introduced by the higher-probability ``small errors'' in its $\numrow$ inputs.
The first term is bounded using noise stability,
and the second is bounded using Lipschitz continuity. 

Second, we instantiate the above insight algorithmically, as does \cite{lmm03}.
Specifically, our algorithm enumerates vectors $\xsparse$ of the desired sparsity
in order to find an approximately optimal solution to our \prob problem.
We note that our guarantees are all parametrized by the Lipschitz continuity $c$
and the noise stability $\beta$ of the function $g$.
Most notably, we obtain an additive polynomial-time approximation scheme (PTAS)
whenever both $\beta$ and $c$ are constants.
% \footnote{Throughout this paper, 
% all approximation guarantees are in the additive sense unless stated otherwise.
% \yucomment{Maybe remove this footnote? The word 'additive' appeared more than 30 times in the paper despite this footnote.}}

Third, we rule out certain natural extensions of our results assuming well-believed complexity-theoretic conjectures. We show that neither Lipschitz continuity nor noise stability alone suffices for an additive PTAS for mixture selection, and both together do not suffice for an additive fully polynomial-time approximation scheme (FPTAS).  For a function which is $O(1)$-stable
yet $O(1)$-Lipschitz continuous only in the $L^1$ metric,
we show approximation hardness by a reduction from the NP-hard maximum independent set problem.
For a function which is $O(1)$-Lipschitz in $L^\infty$
yet not $O(1)$-stable, we show approximation hardness
by a reduction from the planted clique problem. Finally, for a function which is both $O(1)$-Lipschitz in $L^\infty$ and $O(1)$-stable, we rule out an additive FPTAS via a reduction from the maximum independent set problem.

% Third, we examine \prob when $g$ is either Lipschitz continuous or noise stable,
% but not both.
% We show that neither assumption by itself suffices for tractability in general.
% For a function $g$ which is $O(1)$-stable
% yet $O(1)$-Lipschitz continuous only in the $L^1$ metric,
% we show approximation hardness by a reduction from the maximum independent set problem.
% On the other hand, for a function $g$ which is $O(1)$-Lipschitz in $L^\infty$
% yet not $O(1)$-stable, we show approximation hardness
% by a reduction from the planted clique problem.

Despite the simplicity of our framework,
we find that it has powerful implications for problems in mechanism design
and optimal signaling in games.
We feature four natural applications in this paper,
three of which resolve or partially resolve outstanding open problems from prior work:

\begin{enumerate}
\item \textbf{Lottery design}:
Dughmi, Han, and Nisan~\cite{DHN14}
examined one of the most basic problems in mechanism design:
that of designing the revenue-maximizing multi-item auction
for a single unit-demand buyer with valuation represented implicitly via a sampling oracle.
They reduced this problem to a regularized variant of itself,
namely optimally designing a small number of lottery-price pairs ---
a \emph{small menu} --- from which the buyer is allowed to choose.
We apply our framework to resolve the special case of this problem
with a single lottery --- i.e., a menu of size $1$.
This follows from the Lipschitz continuity and  noise stability
of the function
$\gLot{\Weight}(t) := \max\limits_{p}
\Set{p \cdot \sum_{i = 1}^{n} \Weight_i \cdot I[t_i \geq p]}$ for an arbitrary weight vector $w \in \Delta_n$,
where $I[\EEE]$ is the indicator function for the event $\EEE$.
%\eedelete{$\gLot(t) := \max\limits_{p}
%\Set{p \cdot \frac{1}{\numrow}
%\Size{\SpecSet{i \in [n]}{t_i \ge p}}}$.}
\label{problist:lotterydesign}

\item \textbf{Revenue-maximizing signaling in probabilistic second-price auctions}:
Emek et al.~\cite{emeksignaling} and Miltersen and Sheffet~\cite{miltersensignals}
considered signaling in the context of a probabilistic second-price auction.
In particular, the attributes of the item for sale are unknown,
and the auctioneer must decide what information to reveal
in order to maximize his revenue in this auction.
This is particularly relevant in advertising auctions,
where items are impressions associated with demographics that are a-priori unknown
to the advertisers bidding in the auction.
Whereas both papers presented a polynomial-time algorithm for this problem
when bidder types are fixed, the general problem was shown to be NP-hard
and its approximability was left largely open.
Using our framework, a PTAS for the general problem follows easily.
We use the fact that the function $\maxtwo$,
which simply returns the second largest entry of a vector,
is Lipschitz continuous and noise stable.
\label{problist:auctionsignaling}

\item \textbf{Persuasion in voting}:
Alonso and C\^{a}mara~\cite{Alonso14} examine a simple election
for selecting a binary outcome --- say whether a ballot measure is passed ---
when voters are not fully informed of the consequences of the measure,
and hence of their utilities. Each voter casts a Yes/No vote,
and the measure passes if the fraction of Yes votes exceeds a certain pre-specified threshold.
A principal --- say a moderator of a political debate ---
can determine the protocol --- or signaling scheme ---
through which information regarding the measure is gathered and shared with voters.
We consider a principal concerned with maximizing the probability of
the measure passing.
\cite{Alonso14} characterize the optimal signaling scheme
and a number of its properties, though stop short of deriving an algorithm
for optimal signaling.
We design a multi-criteria PTAS for this problem using our framework. Along the way, we also design a bi-criteria PTAS for the related problem of maximizing the expected number of Yes votes in the election.
For both results, we use the fact that the function
$\gVotesum(t) =  \frac{1}{n} \Size{\set{i: t_i \ge 0}}$ 
is noise stable and Lipschitz continuous in a bi-criteria sense.
%\yucomment{The function AC14 used is $ I(\frac{1}{n} \Size{\set{i: t_i \ge 0}} \ge \frac{1}{2})$, or $k/n$ instead of $1/2$ in general. This function is neither O(1)-Lipschitz nor O(1)-stable. I feel the tri-criteria result emerges naturally because we are working with the Lipschitz-extension and ``Stability-extension'' of this function.}
\label{problist:voting}

\item \textbf{Optimal signaling in normal form games}:
Dughmi~\cite{D14} examined the problem of optimal signaling in abstract normal form games,
and ruled out an FPTAS even for two-player zero-sum games.
The possibility for a PTAS for two-player zero-sum games,
and a QPTAS for general games with a constant number of players, were left open.
We show that a bi-criteria QPTAS for normal-form games with a constant number of players
follows from our framework, and applies to a large and natural class of objective functions.
We use the fact that every function is $O(\numrow)$-stable,
and the fact that the function measuring the quality of equilibria satisfies a bi-criteria notion of Lipschitz continuity which we define.
% is Lipschitz continuous in a bi-criteria sense.
\label{problist:gamesignaling}

\end{enumerate}

\subsection*{Additional Discussion of Related Work}

As previously described, our framework generalizes and refines the main insight of \cite{lmm03}. The recent work of Barman~\cite{barman15caratheodory} is also similar in spirit; in particular, the approximate variant of Caratheodory's theorem employed in that paper can be viewed as a mixture selection problem with  $g(Ax) = - ||Ax - Ax^*||_p$ for a fixed vector $x^*$ and norm $p \geq 2$. Even though this function $g$ is neither Lipschitz continuous in $L^\infty$ nor noise stable, Barman exhibits a PTAS under the assumption that the columns of $A$ have small (i.e. constant) $p$-norm.

% We defer most discussion of work related to our applications to the relevant sections. However, we mention that the {\em Fault-Tolerant Distributed Storage} problem studied in \cite{Daskalakis:2014}, and for which they design an EPTAS, can be viewed as a mixture selection problem which is closely related to our voting application.  Our framework yields a related but incomparable result: we remove their ``product distribution'' assumption on the input, but obtain a weaker approximation guarantee (a bi-criteria PTAS). We elaborate on this connection further in Section \ref{sec:voting}.

%%% Local Variables: 
%%% mode: latex
%%% TeX-master: "main"
%%% End: 

%% file: framework.tex
%% Framework

\section{Framework}\label{sec:framework}
%Let $\gdom$ denote the solid dimensional hypercube, normalized to either $[-1,1]^n$ or $[0,1]^n$. 
For a function $g: \gdom \rightarrow [-1, 1]$ and a positive integer $\numcol$,
we define the following optimization problem
which we term $\numcol$-dimensional \emph{\prob} for $g$:
given an $\numrow \times \numcol$ matrix $A$ with entries in $[-1, 1]$,
find $x$ in the $\numcol$-dimensional simplex $\Delta_\numcol$
maximizing $f(x) := g(A x)$.
In this section, we present our notion of \emph{noise stability},
and derive approximation algorithms for this problem
when the function $g$ is simultaneously noise stable
and Lipschitz continuous with respect to the $L^\infty$ metric.
Moreover, we show that neither requirement alone suffices for our results. 

Our approximation guarantees will be additive
--- i.e., an $\epsilon$-approximation algorithm for \prob
outputs $x \in \Delta_\numcol$
with $f(x) \geq \max_{y \in \Delta_\numcol} f(y) - \eps$.
To illustrate our techniques, we use the following function
$\gmid: \gdom \to [-1, 1]$,
which averages all but the top and bottom quartiles of its inputs,
as a running example.
$$\gmid(t) = \frac{1}{\Ceil{3n/4} - \floor{n/4}} 
\sum\limits_{i=\floor{n/4} + 1}^{\Ceil{3n/4}} t_{[i]},$$
where $t_{[i]}$ denote the \Ordinal{i} largest entry of $t$.
Throughout the paper, we use $t_i$ to denote the \Ordinal{i} entry of $t$,
and use $t_{[i]}$ to denote the \Ordinal{i} largest entry of $t$.

Though we present our framework for functions $g:[-1,1]^n \to [-1,1]$,
we define mixture selection similarly for functions $g:[0,1]^n \to [0,1]$. The two definitions are equivalent up to normalization,
and it is easy to verify that all our results and bounds for mixture selection
carry through unchanged to either definition.

%more generally when the domain of $g$ is a product of $n$ closed intervals contained in $[-1,1]$. In this case, we constrain the input matrix $A$ so that each rowwhere $\dom$ defined on an arbitrary closed and convex  domain $\mathtt{dom} \sse [-1,1]^n$,
%in which case the input matrix $A$ is restricted so that its columns lie in $\mathtt{dom}$.

Our main result applies to functions $g$ which are
  both noise stable and Lipschitz continuous with respect to the $L^\infty$ metric. 
We now formalize these two conditions.

\subsection*{Lipschitz Continuity}

A function $g : \gdom \rightarrow [-1, 1]$
is {\em  \cLip continuous in $L^\infty$} --- or \cLip for short ---
if and only if for all $t, t'$ in the domain of $g$,
$\ABS{g(t) - g(t')} \leq c \infNorm{t - t'}$.
To illustrate, our example function $\gmid$ is $1$-Lipschitz.
We note that Lipschitz continuity in $L^\infty$
is a stronger assumption than in any other $L^p$ norm.

\subsection*{Noise Stability}

Our notion of noise stability captures the following desirable property
of a function $g: \gdom \to [-1, 1]$:
if a random process corrupts (i.e., modifies arbitrarily)
some of the inputs to $g$,
with no individual input disproportionately likely to be corrupted,
then the output of $g$ does not decrease by much in expectation.
Such random corruption patterns are captured
by our notion of a \emph{light distribution}
over subsets of $[\numrow]$, defined below.

\begin{definition}[Light Distribution]\label{def:light}
Let \DDD be a distribution supported on subsets of $[\numrow]$.
For $\alpha \in (0,1]$,
we say \DDD is \aLight if and only if
the following holds for all $i \in [\numrow]$:
$\Prob_{R \sim \DDD} [i \in R] \leq \alpha.$
\end{definition}

In other words, a light distribution
bounds the \emph{marginal probability}
of any individual element of $[\numrow]$.
When corrupted inputs follow a light distribution,
no individual input is too likely to be corrupted.
However, we note that our notion of light distribution allows
arbitrary correlations between the corruption events of various inputs.
We define a noise stable function as one which is robust,
in an average sense, to corrupting a subset $R$ of its $\numrow$ inputs
when $R$ follows a light distribution \DDD.
Our notion of robustness is one-sided:
we only require that our function's output not decrease substantially in expectation.
This one-sided guarantee suffices for all our applications,
and is necessitated by some.
We note that the light distribution \DDD,
as well as the (corrupted) inputs,
are chosen adversarially.
We make use of the following notation in our definition:
Given  vectors $t, t' \in \gdom$ and a set $R \sse [\numrow]$,
we say $t' \corr{R}  t$ if $t_i = t'_i$ for all $i \not\in R$.
In other words, if $t' \corr{R} t$,
then $t'$ is a result of corrupting the entries of $t$ corresponding to $R$. 

\begin{definition}[Noise Stability]\label{def:strong-stability}
Given a function $g : \gdom \rightarrow [-1, 1]$
and a real number $\beta \geq 0$, we say $g$ is $\beta$-stable
if and only if the following holds for all $t \in \gdom$,
$\alpha \in (0,1]$, and \aLight distributions \DDD over subsets of $[\numrow]$:
\[\Exp\limits_{R \sim \DDD} \left[
\min\SpecSet{g(t')}{t' \corr{R} t} \right] \geq g(t) - \alpha \beta.\]
\end{definition}

To illustrate this definition,
we show that our example function $\gmid$ is \gStable{4}.
To see this, observe that changing $k$ entries of the input to $\gmid$
can decrease its output by at most $\frac{4k}{n}$.
When $R$ is drawn from an \aLight distribution
and $t$ is an arbitrary input,
$4$-stability therefore follows from the linearity of expectations:
\[
  \Exp_{R \sim \DDD} \left[\min\{\gmid(t') : t' \corr{R} t\}\right] \geq
  \Exp_{R \sim \DDD} \left[ \gmid(t) - \frac{4|R|}{n}\right] \geq \gmid(t) - 4 \alpha.
\]
We note that every function $g:\gdom \to[-1,1]$
  is $2n$-stable, which follows from the union bound.

As a useful building block for proving some of our functions stable,
we show that stable functions can be combined to yield other stable functions
if composed with a convex, nondecreasing, and Lipschitz continuous function.

\begin{prop}\label{prop:close}
Fix $\beta,c \geq 0$,
and let $g_1, g_2,\ldots, g_k: \gdom \rightarrow [-1, 1]$
be $\beta$-stable functions.
For every convex function $h: [-1,1]^k \to [-1,1]$
which is nondecreasing in each of its arguments
and $c$-Lipschitz continuous in $L^\infty$,
the function $g(t) := h(g_1(t),\ldots,g_k(t))$ is $(\beta c)$-stable.
\end{prop}
\begin{proof}\qedsameline
For all $t \in \gdom$ and all \aLight distributions \DDD,
\begin{align*}
\Exp\limits_{R \sim \DDD} [{\min\limits_{t' \corr{R} t} g(t')}]  &=
\Exp\limits_{R \sim \DDD}\ [{\min\limits_{t ' \corr{R} t} h(g_1(t'),\ldots,g_k(t'))}]  \\
&\geq \Exp\limits_{R \sim \DDD}\     
[h(\min\limits_{t' \corr{R} t} g_1(t'), \ldots, \min\limits_{t' \corr{R} t} g_k(t') )]  &
\mbox{(Since $h$ is nondecreasing)} \\
&\geq    h(\Exp\limits_{R \sim \DDD}
[\min\limits_{t' \corr{R} t} g_1(t')], \ldots, 
\Exp\limits_{R \sim \DDD}[\min\limits_{t' \corr{R} t} g_k(t')] )
& \mbox{(Jensen's inequality)} \\
&\geq    h(g_1(t) - \alpha \beta, \ldots, g_k(t) - \alpha \beta)  & 
\mbox{(Stability of each $g_i$)}\\
&\geq   h(g_1(t), \ldots, g_k(t) ) - \alpha \beta c
&\mbox{(Lipschitz continuity of $h$)}\\
&= g(t) - \alpha \beta c.
\end{align*}
\end{proof}

\noindent As a consequence of the above proposition,
a convex combination of $\beta$-stable functions is $\beta$-stable,
and the point-wise maximum of $\beta$-stable functions is $\beta$-stable. 
% As a useful building block for proving some of our functions stable,
%we  that stable functions are closed under convex combinations.
% \begin{prop}\label{prop:close}
% Given two $\beta$-stable functions $g_1, g_2: \gdom \rightarrow [-1, 1]$ and weights $w_1,w_2 \geq 0$ with $w_1+w_2=1$, the function $g(t) = w_1 g_1(t) + w_2 g_2(t)$ is also $\beta$-stable.
% \end{prop}
% \begin{proof}
% For all $t \in \gdom$ and all \aLight distributions \DDD,
% \begin{align*}
% \Exp\limits_{R \sim \DDD} [{\min\limits_{t' \corr{R} t} g(t')}]  &= \Exp\limits_{R \sim \DDD}\ [{\min\limits_{t ' \corr{R} t} (w_1 g_1(t') + w_2 g_2(t'))}] \\
% &\geq \Exp\limits_{R \sim \DDD}\     [{\min\limits_{t ' \corr{R} t} w_1 g_1(t')} +    {\min\limits_{t ' \corr{R} t} w_2 g_2(t')}] 
%  = w_1 \Exp\limits_{R \sim \DDD}\    [{\min\limits_{t ' \corr{R} t} g_1(t')}] +    w_2 \Exp\limits_{R \sim \DDD}\    [{\min\limits_{t ' \corr{R} t} g_2(t')}] \\
%  &\geq w_1 (g_1(t) - \alpha \beta) + w_2 (g_2(t) - \alpha \beta) =    g(t) - \alpha \beta.
% \end{align*}
% \end{proof}

\subsection*{Consequences of Noise Stability and Lipschitz Continuity}

We now state the two main results of our framework.
Both results apply to functions $g: \gdom \to [-1, 1]$
which are simultaneously Lipschitz continuous and noise stable,
and $\numrow \times \numcol$ matrices $A$ with entries in $[-1, 1]$.
Given a vector $x \in \Delta_\numcol$ and integer $s > 0$,
we view $x$ as a probability distribution over $[\numcol]$,
and use the random variable $\tilde{x} \in \Delta_\numcol$
to denote the empirical distribution of $s$ i.i.d.\ samples from $x$.
Formally, $\tilde{x} = \frac{1}{s} \sum_{i = 1}^s e_{k_i}$,
where $k_1,\ldots,k_s \in [\numcol]$ are drawn i.i.d.\ according to $x$,
and $e_j \in \Delta_\numcol$
denotes the \Ordinal{j} standard basis vector.
Since $\tilde{x}$ is the average of $s$ standard basis vectors,
we say it is \emph{$s$-uniform}.

\begin{definition}\label{def:uniform}
We refer to a distribution $y \in \Delta_\numcol$
as \emph{$s$-uniform} if and only if it is the average of a multiset of
$s$ standard basis vectors in $\numcol$-dimensional space.
\end{definition} 

Our first result shows that when the number of samples $s$
is chosen as a suitable function of
the Lipschitz continuity and noise stability parameters,
$g(A \tilde{x})$ is not much smaller than $g(A x)$
in expectation over $\tilde{x}$.
At a high level, we bound this difference as a sum of two error terms:
one accounts for the effect of
low-probability large errors in the inputs $\tilde{t} = A \tilde{x}$ to $g$,
and the other accounts for
effect of higher-probability small errors in the inputs $\tilde{t}$.
The former error term is bounded using noise stability,
and the latter error term is bounded using Lipschitz continuity.

\begin{theorem}
\label{thm:main}
Let $g:\gdom \rightarrow [-1, 1]$ be \abStable and \cLip in $L^\infty$,
let $A$ be an $\numrow \times \numcol$ matrix with entries in $[-1, 1]$,
let $\alpha, \delta > 0$,
and let  $s \geq 2 \ln({\frac{2}{\alpha}}) / \delta^2$ be an integer.
Fix a vector $x \in \Delta_\numcol$,
and let the random variable $\tilde{x}$
denote the empirical distribution of $s$ i.i.d.\ samples
from probability distribution $x$.
The following then holds:
$\Exp[g(A\tilde{x}))] \geq g(A x) - \alpha \beta - c \delta.$
\end{theorem}
\begin{proof}\qedsameline
Denote $t = Ax$ and $\tilde{t} = A \tilde{x}$.
Note that $\tilde{t}$ is a random variable.
Also note that  $t_i$ and $\tilde{t}_i$ can be viewed
as the mean and empirical mean, respectively,
of a distribution supported on $A_{i, 1},\ldots,A_{i, \numcol} \in [-1, 1]$.
We say the \Ordinal{i} entry of $t$ is \emph{approximately preserved}
if $\ABS{t_i - \tilde{t}_i} \leq \delta$,
and we say it is \emph{corrupted} otherwise. Let $R \sse [n]$ denote the set of corrupted entries.
Hoeffding's inequality, and our choice of the number of samples $s$,
imply that $R$ follows an \aLight distribution.

Let $t'$ be such that
  (1) $t'_i = \tilde{t}_i$ for $i \in R$, and 
  (2) $t'_i = t_i$ otherwise. 
Observe that $t' \corr{R} t$, and $\infNorm{t' - \tilde{t}} \leq \delta$.
We can now bound the expected difference between $g(t)$ and $g(\tilde{t})$
as a sum of the error introduced by corrupted entries
and the error introduced by the approximately preserved entries of $t$:
\begin{align*}
  g(t) - \Ex[ g(\tilde{t}) ] =
  \Exp[ g(t) - g(t')] + \Exp[g(t') - g(\tilde{t})] 
  \leq \alpha \beta + c \delta.
\end{align*}
\end{proof}

Notice that if we fix the desired approximation error $\eps$,
the minimum required number of samples~$s$ in Theorem~\ref{thm:main}
to guarantee that $\Exp[g(A\tilde{x}))] \geq g(A x) - \eps$
is obtained by minimizing $\ceil{2 \ln({\frac{2}{\alpha}}) / \delta^2}$
over $\alpha, \delta > 0$ satisfying $\alpha \beta + \delta c \leq \eps$.
Therefore, the required number of samples depends only on the error term $\eps$,
the noise stability parameter~$\beta$,
and the Lipschitz continuity parameter~$c$; in particular,
it is independent of $\numrow$ and $\numcol$.

As a corollary of Theorem \ref{thm:main}, we derive the following algorithmic result.

\begin{theorem}\label{thm:PTAS}
Let  $g:\gdom \to [-1, 1]$ be \abStable and \cLip,
and let $\numcol > 0$ be an integer.
For every $\delta, \alpha > 0$,
the $\numcol$-dimensional \prob problem for $g$
admits an $(\alpha \beta + c \delta)$-approximation algorithm in the additive sense,
with runtime $\numrow \cdot \numcol ^ {O(\log(1/\alpha) / \delta^2)} \cdot T$,
where $T$ denotes the time needed to evaluate $g$ on a single input.
\end{theorem}
\begin{proof}
Let $s \ge 2 \ln(2/\alpha) / \delta^2$ be an integer.
Our algorithm simply enumerates all $s$-uniform distributions,
and outputs the one maximizing $g(A x)$.
This takes time $\numrow \cdot \numcol^{O(s)} \cdot T$.
The approximation guarantee follows from Theorem~\ref{thm:main}
and the probabilistic method.
\end{proof}

As a consequence of Theorem~\ref{thm:PTAS},
the \prob problem for $\gmid$ admits a polynomial-time approximation scheme (PTAS)
in the additive sense.
The same holds for every function $g$
which is $O(1)$-stable and $O(1)$-Lipschitz continuous.
Specifically, by setting
$\alpha = \frac{\eps}{2\beta}$ and $\delta=\frac{\eps}{2c}$,
an $\epsilon$-approximation algorithm runs in time
$n \cdot m ^ {O(c^2 \log(\beta/\eps) / \eps^2)} \cdot T$.
Interestingly, neither noise stability nor Lipschitz continuity alone suffices for such a PTAS,
as we argue in the next subsection.

\subsection*{The Necessity of \emph{Both} Noise Stability and Lipschitz Continuity}

We now present evidence that both our assumptions ---
Noise stability and Lipschitz continuity ---
appear necessary for general positive results
along the lines of those in Theorem~\ref{thm:PTAS}. 

\begin{enumerate}

\item \emph{Stability alone is not sufficient.}
In Section~\ref{sec:hardness:nolip},
we define a function
$\gSlope : [0,1]^\numrow \rightarrow [0, 1]$
which is \gStable{1}.
Furthermore, $\gSlope$ is \Lip{O(1)} with respect to the $L^1$ metric,
which is a weaker property than Lipschitz continuity with respect to $L^\infty$.
We show in Theorem~\ref{thm:gSlope:hard} that
there is a polynomial-time reduction from the maximum independent set problem
on $n$-node graphs
to the $n$-dimensional \prob for $\gSlope$.
Moreover, the reduction precludes a polynomial-time $\epsilon$-approximation algorithm
in the additive sense for some constant $\eps > 0$.

\item \emph{Lipschitz continuity alone is not sufficient.}
One might hope to prove NP-hardness of \prob in the absence of stability.
However, we are out of luck in this regard:
since every function $g: \gdom \to [-1, 1]$ is $2\numrow$-stable,
Theorem~\ref{thm:PTAS} implies a quasipolynomial-time approximation scheme
in the additive sense whenever $g$ is $O(1)$-Lipschitz.
Nevertheless, we prove hardness of approximation
assuming the \emph{planted clique conjecture}
(\cite{jerrum92} and \cite{kucera95}).

More specifically, in Section~\ref{sec:hardness:nostable}
we exhibit a reduction from the planted $k$-clique problem
to \prob for the $3$-Lipschitz function $\gclique_k(t) = t_{[k]} - t_{[k+1]} + t_{[n]}$.
When $k = \omega(\log^2n)$  and
$A$ is the adjacency matrix of an $n$-node undirected graph $G$,
we show that $\max_{x} \gclique_k(A x) \approx 1$ with high probability if $G$ contains a $k$-clique,
and $\max_{x} \gclique_k(A x) \approx \frac{1}{2}$ with high probability
if $G$ is the Erd\"os-R\'enyi random graph $G(n,\frac{1}{2})$.
\end{enumerate}

\subsection*{A Bi-criteria Extension of the Framework}

We have already showed that in the absence of Lipschitz continuity,
one can not hope for a PTAS in general.
Motivated by two of our applications, namely \emph{Optimal signaling in normal form games}
and \emph{Persuasion in voting},
we extend our framework to the design of approximation algorithms for mixture selection
with a \emph{bi-criteria guarantee} when the function in question is stable
but not Lipschitz continuous.
We first define a $(\delta,\rho)$-relaxation of a function.

\begin{definition}\label{def:relaxation}
Given two functions $g,h : \gdom \rightarrow [-1, 1]$ and parameters $\delta,\rho \geq 0$,
we say $h$ is a $(\delta,\rho)$-relaxation of $g$ if for all $t_1, t_2 \in \gdom$ with $\infNorm{t_1 - t_2} \leq \delta$,
$h(t_2) \geq g(t_1) - \rho$.
\end{definition}

In lieu of the Lipschitz continuity condition, we prove our bounds for a relaxation of the function. 

\begin{theorem}
\label{thm:main:bi}
Let $g:\gdom \rightarrow [-1, 1]$ be \abStable,
let $A$ be an $\numrow \times \numcol$ matrix with entries in $[-1, 1]$,
let $\alpha > 0$ and $\delta,\rho \geq 0$,
and let $s \geq 2 \ln({\frac{2}{\alpha}}) / \delta^2$ be an integer.
Fix a vector $x \in \Delta_\numcol$,
and let the random variable $\tilde{x}$
denote the empirical distribution of $s$ i.i.d.\ samples
from probability distribution $x$.
The following then holds for any $(\delta,\rho)$-relaxation $h$ of $g$,
\[\Exp[h(A\tilde{x}))] \geq g(A x) - \alpha \beta - \rho.\]
\end{theorem}
\begin{proof}%\qedsameline
Because the proof is almost identical to the proof of Theorem~\ref{thm:main},
  we just mention the necessary modifications.
Again, let $t=Ax$, let $\tilde{t}=A\tilde{x}$, let $R \sse [n]$ denote the set of corrupted inputs, and let $t'$ be such that $t'_i = \tilde{t}_i$ for $i \in R$ 
  and $t'_i = t_i$ otherwise. Then
\begin{align*}
 g(t) - \Ex[ h(\tilde{t}) ]  &=\Exp[ g(t) - g(t')] + \Exp[g(t') - h(\tilde{t})] \\
&\leq \alpha \beta + \Exp[g(t') - h(\tilde{t})]\\
& \leq \alpha \beta + \rho,
\end{align*}
where the first inequality follows by noise stability of $g$, and the last inequality follows from the fact that $h$ is a $(\delta,\rho)$-relaxation of $g$.
\end{proof}

Having replaced Theorem~\ref{thm:main} by Theorem~\ref{thm:main:bi},
a similar computational result as Theorem~\ref{thm:PTAS} can be inferred
in the bi-criteria sense.

%%% Local Variables: 
%%% mode: latex
%%% TeX-master: "main"
%%% End: 

%% file: lottery.tex
%---- macros used by Ehsan in this section:
% \Set, \SpecSet, \Size, \infNorm
% \aLight, \gStable, \lightSym
% \numrow, \numcol, \xopt

% local macros for this section macros:

\section{Lottery Design for Revenue Maximization}\label{sec:lottery}

To illustrate the utility of our framework, we start with a simple but basic open problem in Bayesian mechanism design posed by 
Dughmi, Han, and Nisan \cite{DHN14}.
An instance of the \emph{lottery design} problem is given by a valuation matrix
$A \in [0, 1]^{\numrow \times \numcol}$
and $\numrow$ non-negative weights $\Weight_1, \ldots, \Weight_\numrow$
with $\sum_{i = 1}^{n} w_i = 1$.
Here $\numrow$ denotes the number of buyer types, $\numcol$ denotes the number of items, and $\Weight$ represents a probability distribution over types.
Each $0 \le A_{i, j} \le 1$ is the value of item $j$ to a buyer of type $i$. The goal is to design a single lottery-price pair $(x, p)$, with $x \in \Delta_\numcol$ and $p \geq 0$, so that the expected revenue of the auction which offers the lottery $x$ over items at price $p$ to a buyer with type drawn according to $\Weight$ is maximized.
%with $x$ in the $\numcol$-dimensional simplex $\Delta_\numcol$
%such that the expected revenue of the auction which offers the lottery $x$ over items t
%when assuming the buyer's type is drawn with respect to $\Weight$,
%is maximized.
We assume the buyer is risk neutral, and therefore accepts the offer precisely if his type $i$ satisfies $A_i x \geq p$. Consequently, our goal is to choose $(x,p)$ maximizing
$p \cdot \sum_{i = 1}^{n} (\Weight_i \cdot I[A_i x \geq p])$, where $I[\EEE]$ denotes the indicator function for the event~$\EEE$.

The lottery design problem is closely related to the general unit-demand single-buyer 
mechanism design problem considered in \cite{DHN14},
where the buyer's type is drawn from
a common knowledge prior distribution \BBB
given by a sampling oracle, and the buyer is to be presented with a \emph{menu} consisting of several lottery-price pairs from which to choose.
\cite{DHN14} frame this mechanism design problem
as a computational task of ``learning'' a good mechanism
by sampling from \BBB, and use the
size of the menu as a regularization constraint in order to  prevent over-fitting the mechanism to the sampled data. 
The problem of maximizing the expected revenue by using a menu of at most a given size is the main algorithmic question in \cite{DHN14}, and its computational complexity is left largely open.
When constrained to a menu with a single lottery, the goal is to choose $(x,p)$ maximizing the expected revenue
$p \cdot \Prob\limits_{a \sim \BBB}[a x \ge p]$.
%This single-menu mechanism design problem
%is a first step to answering the main open question from \cite{DHN14}.

Using our mixture selection framework, we first give an additive PTAS
for the lottery design problem
when the value distribution \BBB is given explicitly by the matrix $A$ and weights $w$ as described above.
We then extend our result to cases in which \BBB can only be accessed through sampling, using fairly standard uniform convergence arguments.
In Section~\ref{sec:hardness:lottery} (Theorem~\ref{thm:lottery:hard}),
we rule out an additive FPTAS for this problem, and in doing so provide complexity-theoretic evidence that our PTAS --- for both lottery design in particular and mixture selection for stable and Lipschitz-continuous functions more generally --- is  essentially the best we can hope for.

\subsection{Lottery Design in the Explicit Input Model}

Given the number of buyer types $n$ and weights $w \in \Delta_n$,  the lottery design problem is simply mixture selection for the function $\gLot{\Weight}(t): [0,1]^n \to [0,1]$, defined below.

\begin{equation}
  \label{eq:gLot}
\gLot{\Weight}(t) = \max\limits_{p \geq 0}\Set{p \cdot \sum_{i = 1}^{n} (\Weight_i \cdot I[t_i \geq p])}  
\end{equation}

%
%
% In the \LOT problem, a matrix $A$ of size $\numrow \times \numcol$
% is given where $\numrow$ is the number of buyer types
% and $\numcol$ is the number of items.
% Here, $0 \le A_{i, j} \le 1$ denotes the value of item $j$
% for a buyer of type $i$.
% We assume that a buyer is uniformly randomly drawn
% from one of the $\numcol$ types.
% Our goal is to choose a lottery-price pair $(x, p)$
% with $x$ in $\numcol$-dimensional simplex $\Delta_\numcol$
% such that the revenue
% $p \cdot \frac{1}{\numrow}\Size{\SpecSet{i}{A_i x \ge p}}$
% is maximized.
% Thus, we define the following objective function:
% $\gLot(t) := \max\limits_{p}
% \Set{p \cdot \frac{1}{\numrow}
% \Size{\SpecSet{i \in [n]}{t_i x \ge p}}}$.
%
As the first step to applying our framework, we show that $\gLot{\Weight}$ is noise stable. The high level idea is the following:  if a subset of the  inputs to $\gLot{\Weight}$ is corrupted, then in the worst case each such input exceeded the price $p$ before corruption but not after corruption. This reduces the output of $\gLot{\Weight}$ by at most the total weight of corrupted inputs. When corrupted inputs are chosen according to an \aLight distribution, their expected total weight is bounded by $\lightSym$. %When $t'$ be obtained from $t$ by corrupting entries corresponding to $R \subseteq \gindices$, it is easy to see that $\gLot{\Weight}(t) \geq \gLot{\Weight}(t') - \Weight(R)$ where $\Weight(R) = \sum_{i \in R} w_i$ denotes the total weight of corrupted entries.  As long as corrupted entries are drawn from  \aLight distribution, $\Exp[\Weight(R)] \leq \lightSym$, and therefore the expected loss is at most $\lightSym$. 

\begin{lemma}\label{lem:gLot_stable} 
The function $\gLot{\Weight}$ is \gStable{1}.
\end{lemma}
\begin{proof}
Let $t \in [0,1]^n$ be an arbitrary input to $\gLot{\Weight}$. When $t'$ is obtained from $t$
by corrupting the entries corresponding to $R \subseteq \gindices$, an event we denote by $t' \corr{R} t$, it is easy to see that $\gLot{\Weight}(t') \geq \gLot{\Weight}(t) - \Weight(R)$
where $\Weight(R) = \sum_{i \in R} w_i$
denotes the total weight of corrupted entries. Moreover, when $R$ is a random variable  drawn from an \aLight distribution $\D$, we can bound the expected loss: $\Exp[\Weight(R)] = \sum_{i=1}^n \Pr[i \in R] \cdot w_i \leq \sum_i \lightSym w_i = \lightSym$. It follows that $\gLot{\Weight}$ is \gStable{1}.
\begin{align*}
\Exp\limits_{R \sim \DDD}
\left[ \min \SpecSet{\gLot{\Weight}(t')}{t' \corr{R} t} \right] &\ge \gLot{\Weight}(t) - \Ex_{R \sim \D}[\Weight(R)] \\ 
&\ge \gLot{\Weight}(t) - \lightSym. \qedhere 
\end{align*}
\end{proof}

% Let $t'$ be obtained from $t$
% by corrupting entries corresponding to $R \subseteq \gindices$.
% Then, $\gLot{\Weight}(t) \geq \gLot{\Weight}(t') - \Weight(R)$
% where $\Weight(R) = \sum_{i \in R} w_i$
% denotes the total weight of corrupted entries. 
% As long as corrupted entries are chosen with respect to an \aLight distribution,
% $\Exp[\Weight(R)] \leq \lightSym$
% and hence the expected loss is at most $\lightSym$. 

% \begin{lemma}
% Function $\gLot{\Weight}$ is \gStable{1}.
% \end{lemma}
% \begin{proof}
% Fix $t$,
% and let \DDD be an \aLight distribution over $[\numrow]$.
% If $R \sim \DDD$ and $t' \corr{R} t$,
% $\gLot{\Weight}(t') \ge \gLot{\Weight}(t) - \Weight(R)$.
% Since $\Exp\limits_{R \sim \DDD}\Weight(R) \le \lightSym$,
% $\Exp\limits_{R \sim \DDD}
% \left[ \min \SpecSet{\gLot{\Weight}(t')}{t' \corr{R} t} \right]
% \ge \gLot{\Weight}(t) - \lightSym$.
% \end{proof}

Next, we prove  Lipschitz continuity. The high-level idea is the following: if all inputs to $\gLot{\Weight}$ decrease by $\delta$, then we need only decrease the price $p$ from expression \eqref{eq:gLot} by $\delta$.
The (weighted) fraction of inputs exceeding the price is at least the same as before.

\begin{lemma}\label{lem:gLot_lipschitz}
The function $\gLot{\Weight}$ is \Lip{1} continuous. 
\end{lemma}
\begin{proof}
Consider $t,t' \in [0,1]^n$ with $\infNorm{t' - t} \leq \delta$. We prove that $\gLot{\Weight}(t') \geq \gLot{\Weight}(t) - \delta$,
and by symmetry it follows that $\gLot{\Weight}(t) \geq \gLot{\Weight}(t') - \delta$. Let $p$ be the optimal price for $t$ --- i.e. the maximizer of expression \eqref{eq:gLot} --- and define $p' = \max\Set{0, p - \delta}$. Whenever $t_i \geq p$ we have $t'_i \geq p'$, and   therefore $\sum\limits_{i = 1}^{n} (\Weight_i \cdot I[t'_i \geq p']) \geq
\sum\limits_{i = 1}^{n} (\Weight_i \cdot I[t_i \geq p])$. It follows that $\gLot{\Weight}(t') \geq \gLot{\Weight}(t) - \delta$.
\end{proof}
% \begin{proof}
% Consider $t_1$ and $t_2$ and let $\delta = \infNorm{t_1 - t_2}$.
% We prove that $\gLot{\Weight}(t_2) \geq \gLot{\Weight}(t_1) - \delta$,
% and by symmetry, we will also have $\gLot{\Weight}(t_1) \geq \gLot{\Weight}(t_2) - \delta$.
% Suppose the optimal price for $t_1$ is $p_1$.
% Let $p_2 = \max\Set{0, p_1 - \delta}$.
% Then, $\sum\limits_{i = 1}^{n} (\Weight_i \cdot I[(t_2)_i \geq p_2]) \geq
% \sum\limits_{i = 1}^{n} (\Weight_i \cdot I[(t_1)_i \geq p_1])$,
% just because $t_1$ and $t_2$ are $\delta$-close
% with respect to $L^\infty$.
% Therefore, $\gLot{\Weight}(t_2) \geq \gLot{\Weight}(t_1) - \delta$.
% \end{proof}

Combining Lemmas \ref{lem:gLot_stable} and \ref{lem:gLot_lipschitz} with 
Theorem~\ref{thm:PTAS} yields an additive PTAS
for lottery design.
\begin{theorem}\label{thm:lottery:PTAS}
There is an additive PTAS for the lottery design problem when the valuation distribution is given explicitly by a matrix $A \in [0,1]^{n \times m}$ and a weight vector $\Weight \in \Delta_n$. %The runtime is polynomial in $n$ and 
%Given an $\numrow \times \numcol$ matrix $A$ with entries in $[0,1]$,
%there exists an additive PTAS for maximizing $\gLot{\Weight}(A x)$
%over $x \in \Delta_\numcol$.
\end{theorem}

\subsection{Lottery Design in the Sample Oracle Model}\label{sec:lottery:oracle}

So far, we have assumed that the buyer's type distribution is given explicitly.
We now show how to extend our results to the \emph{sample oracle model}.
Specifically, we assume the buyer's type $a \in [0,1]^\numcol$
is drawn from a distribution \BBB given by a sampling oracle, and seek a randomized approximation scheme with runtime (and number of samples) polynomial in $m$ for each desired approximation guarantee $\eps$.
To simplify exposition
we assume \BBB has finite support, though our results hold more generally.
As usual, our goal is to choose a lottery $x \in \Delta_\numcol$
and a price $p\geq 0$
to maximize $\Rev_\BBB(x,p)= p \cdot \Prob\limits_{a \sim \B} [a x \geq p]$.
%Let $\Rev(x) := \max\limits_{p \geq 0}
%\Set{p \cdot \Prob_{a \sim \B} [a x \ge p]}$.

%We now consider the more general setting,
%  where the buyer type is drawn from an arbitrary distribution \BBB
%  that is only accessible via ``black box'' query.
%In this case, our goal is to maximize
%  $p \cdot \Prob\limits_{a \sim \BBB}[a x \ge p]$.
%Instead of having explicit description of \BBB,
%  we only assume oracle access to sample from \BBB.

\begin{theorem}
\label{thm:lottery_pras}
There is an additive polynomial-time randomized approximation scheme (PRAS)
for the lottery design problem in the sample oracle model.
\end{theorem}
\begin{proof}
Recall that the PTAS in Theorem~\ref{thm:lottery:PTAS} optimizes over all $s$-uniform $m$-dimensional lotteries, where $s$ depends only on the desired approximation guarantee $\epsilon > 0$, and is bounded by a polynomial in $\frac{1}{\eps}$. In particular, $s$ is independent of the number of types $\numrow$, implying that the same approach of enumerating all $s$-uniform lotteries $\tilde{x}$ would succeed in the sample oracle model were it not for our inability to evaluate $\max_p \Rev_\BBB(\tilde{x},p)$  exactly.

We overcome this difficulty by Monte Carlo sampling from \BBB. Given $n$ types $a_1,\ldots,a_n \in [0,1]^m$ sampled from \BBB and presented as the rows of a matrix $A \in [0,1]^{n \times m}$, we run the PTAS from Theorem~\ref{thm:lottery:PTAS} with approximation parameter $\eps$ on the empirical distribution given by $A$ and uniform weights $w_i=\frac{1}{n}$ for all $i$. Taking $n$ to be a suitable polynomial in $m$, $\frac{1}{\eps}$, and $\log(\frac{1}{\gamma})$ --- where $\gamma > 0$ is a parameter --- guarantees that $|\Rev_\BBB(\tilde{x},p) -  p\cdot \frac{1}{n} \sum_{i=1}^n I[A_i \tilde{x} \geq p]| \leq \eps$ simultaneously for all $s$-uniform lotteries $\tilde{x}$ and prices $p$ with probability at least $1-\gamma$.  This follows from standard tail bounds and the union bound, coupled with a uniform convergence argument over prices $p \in [0,1]$. Therefore, with probability $1-\gamma$ our algorithm outputs a lottery-price pair whose expected revenue for a buyer drawn from \BBB is within $O(\eps)$ from the optimal.
\end{proof}

% \sdcomment{Is the above true without discretizing prices? I think so, but it is not immediate. Nevertheless, my feeling is that it's OK to not provide any more detail.}

% \begin{proof}
% Recall that the number of samples $s$ stated in Theorem~\ref{thm:main}
% is independent of $\numrow$, 
% and therefore there exists $\xsparse$ which is $s$-uniform for
% such that $\Rev(\xsparse, \popt) \ge OPT - \eps$
% where $OPT$ is the maximum possible revenue and $p*$ is the corresponding optimal price.
% The problem is that when \BBB is not explicitly given, 
% so exact evaluation of $g$ is impossible.
% However, if $g$ can be evaluated with small additive error using a valuation matrix $A$,
% then the algorithm still works.

% The idea is to use samples from \BBB
% to \emph{estimate} the revenue that each $O(1)$-uniform $\xsparse$ achieves.
% It is easy to see that polynomially-many samples from \BBB
% suffices for close estimation for all sparse $\xsparse$ \emph{simultaneously}.
% This is essentially the standard union-bound and uniform convergence argument.
% So with high probability over the samples from \BBB,
% we obtain accurate estimate of $\Rev(\xsparse)$
% for all $s$-uniform vectors $\xsparse$.

% Our PRAS first samples a polynomial-size matrix $A$
% using the sampling oracle for \BBB
% and then solves the problem $A$ using the PTAS in Theorem~\ref{thm:lottery:PTAS}. 
% Notice that the sampling procedure introduces an extra additive error $\epsilon$,
% but it can be absorbed into our guarantee.
% \end{proof}

%%% Local Variables: 
%%% mode: latex
%%% TeX-master: "main"
%%% End: 

%% file: signaling.tex
\section{Signaling}
\label{sec:signaling}

In the next few sections, we consider a number of Bayesian games
  in which a key parameter $\theta$, the \emph{state of nature},
  in part determines the payoff structure of the game.
We use $\Theta$ to denote the set of all states of nature,
  and assume $\theta \in \Theta$ is drawn from a
  common-knowledge prior distribution which we denote by $\lambda$.
In all our applications, we assume players a-priori know nothing
  about $\theta$ other than its prior distribution $\lambda$,
  and examine policies whereby a principal with access to
  the realized value of $\theta$ may commit to a policy of
  revealing information to the players regarding $\theta$.
  This is often referred to as \emph{signaling} (see e.g. \cite{emeksignaling, miltersensignals, DIR14, D14}).
We restrict our attention to symmetric signaling schemes,
  in which the principal must reveal the same information to all players in the game. 
Thus, a \emph{symmetric signaling scheme} is given by a set $\Sigma$ of 
  \emph{signals}, and a (possibly randomized) map $\varphi$ 
  from states of nature $\Theta$ to signals $\Sigma$. 
The goal of the principal, who is privy to confidential
  state-of-nature information, is to boost her own objective
  by using a signaling scheme 
  $\varphi$ that optimally affects the outcome of the game.

In this section, in addition to providing the technical background for signaling schemes, 
  we~use our framework to define an abstract signaling problem and characterize
  its approximation complexity.
This abstract problem captures the essence of all signaling problems
   considered in this~paper.

\subsection{Background: Signaling Schemes} \label{prelim:schemes}

Let $m = |\Theta|$. Abusing notation, we use $\varphi(\theta,\sigma)$ to denote 
  the probability of announcing signal $\sigma \in \Sigma$ 
  conditioned on the state of nature is $\theta \in \Theta$.
It is well known (\cite{bpbp,D14}) that
  signaling schemes are in one-to-one correspondence with
  \emph{convex decompositions} of the prior distribution $\lambda \in
  \Delta_m$:
Formally, a signaling scheme $\varphi: \Theta \to \Sigma$ 
  corresponds to the convex decomposition  
$\lambda = \sum_{\sigma \in \Sigma} \nu_\sigma \cdot  \mu_\sigma,$
where (1) $\nu_\sigma= \Pr_{\theta\sim\Theta} [\varphi(\theta) = \sigma] = \sum_{\theta
  \in \Theta} \lambda(\theta) \varphi(\theta,\sigma)$  
 is the probability of announcing signal $\sigma$,
and
(2) $\mu_\sigma(\theta) = \Pr_{\theta\sim\Theta} [ \theta | \varphi(\theta) = \sigma] =
  \frac{\lambda(\theta) \varphi(\theta,\sigma)}{\nu_\sigma}$
is the \emph{posterior belief distribution} 
  of $\theta$ conditioned on signal $\sigma$.
The converse is also true: every convex decomposition of $\lambda\in \Delta_m$
  corresponds to a signaling scheme.
Alternatively, the reader can view a signaling scheme $\varphi$
  as the $m\times |\Sigma|$ matrix of pairwise probabilities
  $\varphi(\theta,\sigma)$ satisfying conditions (1) and (2) with respect to 
  $\lambda\in \Delta_m$.

Note that each posterior distribution $\mu \in \Delta_m$
  defines a Bayesian game, and the principal's utility depends on the  outcome of the game. Given a suitable equilibrium concept and selection rule, we let $f: \Delta_m \to \mathbb{R}$ denote the principal's utility as a function of the posterior distribution~$\mu$.
For example, in an auction game $f(\mu)$ may be the social welfare or principal's revenue 
  at the induced equilibrium, or any weighted combination of players' utilities,
  or something else entirely.
The principal's objective as a function of the signaling scheme $\varphi$  can be mathematically expressed by
$F(\varphi,\lambda) = \sum_\sigma \nu_\sigma \cdot f(\mu_\sigma)$.

In this setup, the optimal choice of a signaling scheme is related 
  to the \emph{concave envelope} $f^+$ of the function $f$ (\cite{bpbp,D14}).\footnote{$f^+$ is the point-wise lowest concave function $h$
  for which $h(x) \geq f(x)$ for all $x$ in the domain. 
  Equivalently, the hypograph of $f^+$ is
  the convex hull of the hypograph of $f$.}
Specifically, such a signaling scheme 
  achieves $\sum_\sigma \nu_\sigma \cdot f(\mu_\sigma) = f^+(\lambda)$. 
Thus, there exists a signaling scheme with
  $m+1$ signals that maximizes the principal's
  objective, by applying Caratheodory's theorem to the hypograph of $f$.

\subsection{An Abstract Signaling Problem and its Polynomial-Time Approximation}

% To connect with our mixture-selection framework, 
%   we now consider the following abstract signaling problem:
% Given a function $g$ from the solid
%   $n$-dimensional hypercube $[-1,1]^n$ to the bounded interval $[-1,1]$, an
%   $n \times m$ matrix $A \in [-1,1]^{n\times m}$,
% and a prior distribution $\lambda \in \Delta_m$, 
%   find a signaling scheme $\varphi$
%   that optimizes the principal's objective in expectation.

To connect to our mixture selection framework, we consider signaling problems in which the principal's utility $f(\mu)$ from a posterior distribution $\mu \in \Delta_m$ can be written as $g(A \mu)$ for a function $g: [-1,1]^n \to [-1,1]$ and a matrix $A \in [-1,1]^{n \times m}$.
As described in Section \ref{prelim:schemes}, a signaling scheme $\varphi$ with signals $\Sigma$ corresponds to a family of probability-posterior pairs $\{(\nu_\sigma,\mu_\sigma)\}_{\sigma\in\Sigma}$ decomposing the prior $\lambda \in \Delta_m$ into a convex combination of posterior distributions (one per signal): $\lambda = \sum_{\sigma \in \Sigma} \nu_\sigma \mu_\sigma$. The objective of our signaling problem is then
\begin{equation*}
F(\varphi) = \sum_{\sigma \in \Sigma} \nu_\sigma f(\mu_\sigma)
  = \sum_{\sigma \in \Sigma} \nu_\sigma g(A \mu_\sigma).
\end {equation*}

%Using $\{(\nu_\sigma,\mu_\sigma)\}_{\sigma\in\Sigma}$   as our variables, 

We note that this signaling problem can alternatively be written as
  an (infinite-dimensional) linear program which searches over  probability measures supported on $\Delta_m$ with expectation $\lambda$. The separation oracle for the dual of this linear program is a \prob problem. Whereas we do not use this infinite-dimensional formulation nor its dual directly, we nevertheless show that the same conditions --- noise stability and Lipschitz continuity --- on the function $g$ which lead to an approximation scheme for \prob also lead to a similar approximation scheme for our signaling problem with $f(\mu) = g(A \mu)$.% for $g$ problem also lead to an approximation scheme for for mixture selection also lead to an approximation scheme f PTAS for mixture selection for $g$ leads to a PTAS for the associated signaling problem with $f(\mu) = g(A \mu)$.

%Using the following lemma, we show that the same stability condition
%  that facilitates the PTAS for the \prob problem (the dual) leads
%  to an efficient algorithm for the signaling problem (the primal).

\begin{lemma}%[Sparse Posteriors]
If $g$ is $\beta$-stable and $c$-Lipschitz,
  then for any constants $\alpha,\delta > 0$,
  and for any integer $s \ge 2 \delta^{-2} \ln(2/\alpha)$,
  there exists a signaling scheme $\tilde{\varphi}$ %$\varphi = \{(\nu_\sigma,\mu_\sigma)\}_{\sigma\in\Sigma}$  
for which every posterior distribution is $s$-uniform, and $F(\tilde{\varphi}) \geq OPT - (\alpha \beta + c \delta)$ where OPT denotes the value of the optimal signaling scheme.%  which uses only $s$-uniform sparse posteriors,   and achieves an objective value close to the optimal signaling scheme,   with additive loss $\alpha\beta+c\delta$.
\label{lem:sparse_signals}
\end{lemma}
\begin{proof}
Let $s \geq 2 \delta^{-2} \ln(2/\alpha)$, and let $\tau \in [m^s]$ index all $s$-uniform posteriors, with $\tilde{\mu}_\tau$ denoting the $\tau$'th such posterior. For an arbitrary signaling scheme $\varphi = (\Sigma, \{(\nu_\sigma,\mu_\sigma)\}_{\sigma\in\Sigma})$, we show that each posterior $\mu_\sigma$ can be decomposed into $s$-uniform posteriors without degrading the objective by more than $\alpha \beta + c \delta$; more formally:
\begin{enumerate}
\item 
$\mu_\sigma$ can be expressed as a convex combination of $s$-uniform posteriors as follows.

\begin{equation}
  \label{eq:condition1}
\mu_\sigma = \sum_{\tau \in [m^s]} \tilde{\nu}_{\sigma,\tau} \tilde{\mu}_\tau
\text{\quad with \quad }  \tilde{\nu}_{\sigma} \in \Delta_{m^s}.
\end{equation}
\item
The value of objective function, i.e., $g(A\mu_\sigma)$, is decreased by no more than $\alpha\beta+c\delta$ through this decomposition,
\begin{equation}
  \label{eq:condition2}
\sum_{\tau \in [m^s]} \tilde{\nu}_{\sigma,\tau} \cdot g(A \tilde{\mu}_\tau)
\ge g(A\mu_\sigma) - (\alpha\beta+c\delta).
\end{equation}
\end{enumerate}
%The existence of such a decomposition follows from Theorem~\ref{thm:main} by taking $\tilde{\nu}_\sigma \in \Delta_{m^s}$ to be the probability distribution of $s$ i.i.d samples from $\mu_\sigma$.
The existence of such a decomposition follows from Theorem~\ref{thm:main}: Fix $\sigma$, and let $\tilde{\mu} \in \Delta_m$ be the empirical distribution of $s$ i.i.d. samples from distribution $\mu_\sigma \in \Delta_m$.  The vector $\tilde{\mu}$ is itself a random variable supported on $s$-uniform posteriors, its expectation is $\mu_\sigma$, and by Theorem~\ref{thm:main} we have $\ex[g(A \tilde{\mu})] \ge g(A \mu_\sigma) - (\alpha\beta+c\delta)$.  Therefore, by taking $\tilde{\nu}_{\sigma,\tau} = \Pr [ \tilde{\mu} = \tilde{\mu}_\tau]$ for each $\tau \in [m^s]$ we get the desired decomposition of $\mu_\sigma$.

The lemma follows by composing the decomposition $\varphi$ with the decompositions of the posterior beliefs $\mu_\sigma$ to yield a signaling scheme $\tilde{\varphi}$ with only $s$-uniform posteriors and $F(\tilde{\varphi}) \geq F(\varphi) - (\alpha \beta + c \delta)$.
Specifically, the signals of $\tilde{\varphi}$ are $\Sigma \cross [m^s]$, where  signal $(\sigma, \tau)$ has probability $\nu_\sigma \cdot \tilde{\nu}_{\sigma,\tau}$ and induces the posterior $\tilde{\mu}_\tau$. \footnote{Note, however, that we can also ``merge'' all signals with the same posterior $\tilde{\mu}_\tau$ without loss.}
Using Equation \eqref{eq:condition1} and \eqref{eq:condition2}, it is easy to verify that this describes a valid signaling scheme with $F(\tilde{\varphi}) \geq F(\varphi) - (\alpha \beta + c \delta)$.
%
%\[ \sum_{(\sigma, \tau)}  \nu_\sigma \tilde{\nu}_{\sigma,\tau} g(\tilde{\mu}_\tau) \geq \sum_{\sigma} \nu_\sigma \]
%Specifically,  from the fact that the composition of the decomposition $\varphi$ with the above described decomposition of each $\mu_\sigma$ describe
%
\end{proof}

Lemma \ref{lem:sparse_signals} permits us to restrict attention to $s$-uniform posteriors without much loss in our objective. Since there are only $m^s$ such posteriors,  a simple linear program with $m^s$ variables computes an approximately optimal  signaling scheme.
%  enables  us to focus on ``uniform'' sparse 
 % posteriors which we can efficiently enumerate.
%There are only $|P_s| = m^s$ such posteriors,
 % so we can compute the optimal signaling scheme over them in $\poly(m^s)$ time.

\begin{theorem}[Polynomial-Time Signaling]
\label{thm:signal}
If $g$ is $\beta$-stable and $c$-Lipschitz, then
  for any constant $\alpha,\delta > 0$,
  there exists a deterministic algorithm that constructs a signaling scheme
  with objective value at least $OPT - (\alpha\beta+c\delta)$,
  where $OPT$ is the value of the optimal signaling scheme.
Moreover, the algorithm runs in time
  $\poly(m^{ \delta^{-2} \ln(1/\alpha)})\cdot n \cdot T$, where $T$ denotes the time needed to evaluate $g$ on a single input.
\end{theorem}
\begin{proof}
Let $s$ be an integer with $s \geq (2 \delta^{-2} \ln(2/\alpha))$, denote $M=m^s$, and let $\set{\mu_1, \cdots, \mu_M}$ be the set of all $s$-uniform posteriors.
Lemma~\ref{lem:sparse_signals} shows that restricting to 
$s$-uniform posteriors only introduces an $\alpha\beta+c\delta$ additive loss in the objective.
Thus it suffices to compute the optimal signaling scheme supported only on $s$-uniform posteriors.  
This can be done using the following linear program:
\begin{lp}
\label{lp:main}
\maxi{\sum_{j\in[M]} \nu_j \cdot g(A \mu_j)}
\st \con{\sum_{j\in[M]} \nu_j \mu_j = \lambda}
    \con{\nu \in \Delta_M} 
\end{lp}
Note $\mu_j$ is the $j$'th $s$-uniform posterior ---
  the only variables in this LP are $\nu_1, \ldots, \nu_M$.
\end{proof}

Our proofs can be adapted to obtain a bi-criteria guarantee in the absence of Lipschitz continuity, as in Section \ref{sec:framework}. The following Theorem follows easily, and we omit the details.

\begin{theorem}[Polynomial-Time Signaling (Bi-criteria)]
\label{thm:signal:bicriteria}
Let $g,h: [-1,1]^n \to [-1,1]$ be such that $g$ is $\beta$-stable and $h$ is a $(\delta,\rho)$-relaxation of $g$, and let $\alpha >0$ be a parameter. There exists a deterministic algorithm which, when given as input a matrix $A \in [-1,1]^{n \times m}$ and a prior distribution $\lambda \in \Delta_m$, constructs a signaling scheme $\varphi= \{(\nu_\sigma,\mu_\sigma)\}_{\sigma\in\Sigma}$ such that \[\sum_{\sigma \in \Sigma} \nu_\sigma h(A \mu_\sigma) \geq OPT - \alpha \beta - \rho,\]
  where $OPT$ is the maximizer of $F(\varphi^*) = \sum_{\sigma \in \Sigma^*} \nu^*_\sigma g(A \mu^*_\sigma)$ over signaling schemes $\varphi^*= \{(\nu^*_\sigma,\mu^*_\sigma)\}_{\sigma\in\Sigma^*}$. 
Moreover, the algorithm runs in time
  $\poly(m^{ \delta^{-2} \ln(1/\alpha)})\cdot n \cdot T$, where $T$ denotes the time needed to evaluate $h$ on a single input.
\end{theorem}

\paragraph{Remarks} We note that our proof suggests an extension of the result in Theorem \ref{thm:signal} 
%\ref{thm:signal:bicriteria}
to cases in which $f$ is given by a ``black box'' oracle, so long as we are promised that it is of the form  $f(\mu) = g(A\mu)$. In this model the runtime of our algorithm does not depend on $n$, but instead depends on the cost of querying $f$. We also point out that
  that even though we precompute the quality of all $m^s$
  posteriors, we can guarantee that our output signaling scheme uses at most $m+1$ signals; this is because  LP \eqref{lp:main} has only $m+1$ constraints, and therefore admits  an optimal solution where at most $m+1$ variables are non-zero.

%%% Local Variables: 
%%% mode: latex
%%% TeX-master: "main"
%%% End: 

%% file: auction.tex
%% Second-Price Auction

\section{Signaling in Probabilistic Second-Price Auctions}\label{sec:auction}

We examine signaling in probabilistic second-price auctions, as considered by Emek et al.~\cite{emeksignaling} and Miltersen and Sheffet \cite{miltersensignals}. In this setting, the item being auctioned is probabilistic, and the instantiation of the item is known to the auctioneer but not to the bidders. The auctioneer commits to a signaling scheme for (partially) revealing information about the item for sale before subsequently running a second-price auction. We consider a probabilistic second-price auction described by the following parameters:

% Emek et al.~\cite{emeksignaling} and Miltersen and Sheffet \cite{miltersensignals}
%   considered the signaling problem faced by an informed auctioneer who needs
%   to decide how to (partially) reveal information regarding the item for sale,
%   aiming to maximize his revenue.
% A Bayesian single item, second-price auction
%   can be described by the following parameters:
\begin{itemize}
\item An integer $n$ denoting the number of bidders. We index the players by the set $[n]=\set{1,\ldots,n}$.
\item An integer $m$ denoting the number of \emph{states of nature}. We index states of nature by the set $\Theta=\set{1,\ldots,m}$. Each $\theta \in \Theta$ represents a possible instantiation of the item being sold. %   denoted by $\theta$
\item A common-knowledge prior distribution $\lambda \in \Delta_m$ on the states of nature.
\item A common-knowledge prior distribution  $\D$ on \emph{valuation matrices} $\V \in [0,1]^{n \times m}$, given either explicitly or as a ``black-box'' sampling oracle. For a valuation matrix $\V$, entry $\V_{ij}$ denotes the value of player $i$ for the item corresponding to state of nature $j$.  
\end{itemize}

 The game being played is the following: (a) The auctioneer first commits to a signaling scheme $\varphi: \Theta \to \Sigma$; (b) A state of nature  $\theta \in \Theta$ is drawn according to $\lambda$ and revealed to the auctioneer but not the bidders; (c) The auctioneer reveals a public signal $\sigma \sim \varphi(\theta)$ to all the bidders; (d) A valuation matrix $\V \in [0,1]^{n \times m}$ is drawn according to $\D$, and each player $i$ learns his value $\V_{i,j}$ for each potential item $j$; (e) Finally, a second-price auction for the item is run.

As an example, consider an auction for an umbrella: the state of nature $\theta$ can be the weather tomorrow, which determines the utility $\V_{i,\theta}$ of an umbrella to player $i$. 
We assume that $\lambda$ and $\D$ are independent. We also emphasize that a bidder knows nothing about $\theta$ other than its distribution $\lambda$ and the public signal $\sigma$, and the auctioneer knows nothing about $\V$ besides its distribution $\D$ prior to running the auction.

%  The state of nature $\theta \in \Theta$,
%   is first drawn from $\lambda$, and revealed
%   to the auctioneer but not to the bidders. 
% Then the auctioneer reveals a public signal $\sigma$. 
% Subsequently, bidders' valuations $\V \in
%  [0,1]^{n \times m}$ are drawn from $\D$, 
%  where $\V_{i,j}$ is player $i$'s value for the 
%  item when state of nature is $j$, and each
%  player $i$ privately learns this valuation.
% For example, state of nature $\theta$ can be the weather tomorrow, affecting the utility $\V_{i,j}$ of an umbrella for player $i$.
% Finally, a second-price auction is run, where bidders bid according to their
%  private valuations and their posterior belief (after learning
%  $\sigma$) regarding the state of nature.
% We also emphasize that the auctioneer knows nothing regarding 
%   $\V$ besides its distribution $\D$ prior to running the auction, 
%   and that the bidders know nothing regarding $\theta$ 
%   besides its distribution $\lambda$ and the signal $\sigma$.

%In the Bayesian second-price auction,
  We adopt the (unique) dominant-strategy truth-telling 
  equilibrium as our solution concept. 
Specifically, given a signaling scheme $\varphi: \Theta \to \Sigma$ and a signal $\sigma \in \Sigma$, in the subgame corresponding to $\sigma$ it is  a dominant strategy for player $i$ to bid
  $\Ex_{\theta \sim \lambda } [ \V_{i\theta} | \varphi(\theta) = \sigma ]$ --- 
  his posterior expected value for the item conditioned on the received signal~$\sigma$. Therefore the item goes to the  player with maximum posterior expected value, at a price equal to the second-highest  posterior expected value.

The algorithmic problem we consider is the one faced by the auctioneer in step (a) --- namely computing an optimal signaling scheme --- assuming the auctioneer looks to maximize  expected revenue. It was shown in \cite{emeksignaling,miltersensignals} that polynomial-time algorithms exist for several special cases of this problem. However, the general problem was shown to be NP-hard even with 3 bidders --- specifically, no additive FPTAS exists unless P = NP. In this section, we resolve the approximation complexity of this basic signaling problem
  by giving an additive PTAS. We note that variations of this problem were considered in \cite{GhoshNS07,guo}, with different constraints on the signaling scheme ---  the results in these works are not directly relevant to our model.

\subsection{Revenue is Stable}

Given a signaling scheme $\varphi$ expressed as a decomposition
  $\{\nu_\sigma,\mu_\sigma\}_{\sigma\in\Sigma}$ of the prior distribution $\lambda$, we can express the auctioneer's expected revenue as 
\begin{align*}
%\textit{rev}(\nu,\mu,\D)  &= 
\sum_{\sigma \in \Sigma} \nu_\sigma \Ex_{\V\sim\D} \maxtwo (\V  \mu_\sigma),
\end{align*}
where the function $\maxtwo$ returns the second largest entry of a given vector, i.e. $\maxtwo(t) = t_{[2]}$.
To apply our main theorem, we need to show that the revenue in a subgame with posterior distribution $\mu \in \Delta_m$ --- namely $\Ex_{\V\sim\D} \maxtwo (\V  \mu)$ --- can be written in the form $g(W \mu)$ for a matrix $W$. To facilitate our discussion we assume that the valuation distribution $\D$ has finite support size $C$, though this is without loss of generality.
Imagine we form a large matrix $W$ by stacking matrices in the support of $\D$ on top of each other.
Formally, $W = [\V_1^T, \V_2^T, \cdots, \V_C^T]^T$ 
  where $\V_i$ is the $i$th matrix in the support of $\D$.
When matrix $\V_i$ is drawn from $\D$, we take the second-highest bid 
  from the rows of $W$ corresponding to $\V_i$
  (rows $(i - 1) \cdot n + 1$ to $i \cdot n$, where $n$ is the number players).
For $S \sse [nC]$ and $t \in [0,1]^{nC}$, let $\maxtwo_S(t)$ denote the second-highest value among entries of $t$ indexed by $S$. Then we can write the auctioneer's expected revenue as
\[
	\gRev(W \mu) = \Ex_{\V\sim\D} \maxtwo\nolimits_{S(\V)} (W  \mu)
\]
where $S(\V)$ is the set of rows in $W$ corresponding to $\V$.

% Without loss of generality, in our stability analysis below
%   we assume $\D$ has a finite support of size $C$.
% \footnote{
%   We can also handle the case when $\D$ has infinite support.
%   For ease of presentation, we assume $C$ is finite.}
% Given a signaling scheme $\varphi$ for prior $\lambda$,
%   and a valuation matrix $\V \in [0,1]^{n \times m}$,  and by viewing $\varphi$ as decomposition 
%   $\{\nu_\sigma,\mu_\sigma\}_{\sigma\in\Sigma}$ of the prior distribution $\lambda$,
%   the expected revenue is given by
% \begin{align*}
% \textit{rev}(\nu,\mu,\D) 
% &= \sum_{\sigma \in \Sigma} \nu_\sigma \Ex_{\V\sim\D} \maxtwo (\V  \mu_\sigma),
% \end{align*}
% where the function $\maxtwo$ returns the second largest entry of a given vector, i.e. $\maxtwo(t) = t_{[2]}$.
% To apply our main theorem, we need to write this in the form $g(A x)$.
% Imagine we form a large matrix $W$ by stacking matrices in the support of $\D$ on top of each other.
% Formally, $W = [\V_1^T, \V_2^T, \cdots, \V_C^T]^T$ 
%   where $\V_i$ is the $i$'th matrix in the support.
% We now take the second highest bid 
%   from the corresponding rows of $\V_i$ in $W$ 
%   (formally, these are rows $(i - 1) n + 1$ to $i \cdot n$ where $n$ is the number players).
% Let $\maxtwo_S(t)$ to be the function that
%   takes the second highest value among entries of $t$ indexed by $S$. Then,
% \[
% 	rev(W \mu) = \Ex_{\V\sim\D} \maxtwo\nolimits_{S(\V)} (W  \mu)
% \]
% where $S(\V)$ is the set of corresponding rows of $\V$ in $W$. 

\begin{lemma}[Smooth and Stable Revenue]
The function $\gRev(t) = \Ex_{\V\sim\D} \maxtwo\nolimits_{S(\V)} (t)$ is 1-Lipschitz and 2-stable.
\label{lem:rev_stable}
\end{lemma}
\begin{proof}
Because $\maxtwo_S$ is 1-Lipschitz for a fixed set of indices $S$, it follows that $\gRev$, which is  a convex combination of
these 1-Lipschitz functions, is also 1-Lipschitz.

To show that $\gRev$ is stable, we first show that the function $\maxtwo:[0,1]^n \to [0,1]$ is
stable. Given $t \in [0,1]^n$ and a random set  $R \sse [n]$ drawn from an $\alpha$-light
distribution $\D$, the union bound implies that $R$ includes neither of the two largest entries of $t$
with probability at least $1-2\alpha$. In this case, the value of $\maxtwo$ is not
affected by corruption of the entries indexed by $R$. Hence 
\[ 
\Ex_{R\sim\D}\ \left[\min\{\maxtwo(t'):t' \corr{R} t \}\right] 
\ge (1-2\alpha) \cdot \maxtwo(t) + 2\alpha \cdot 0
\geq \maxtwo(t) - 2\alpha.
\]
Therefore $\maxtwo$ is 2-stable, which implies that $\maxtwo_S: [0,1]^{nC} \to [0,1]$ is also
2-stable for any fixed set of indices $S$. The function $\gRev$ is a convex combination of functions of the form $\maxtwo_S$, and is therefore also 
2-stable by Proposition~\ref{prop:close}.
\end{proof}

\begin{theorem}
\label{thm:auction_ptas}
The revenue-maximizing signaling problem in probabilistic second-price auctions admits an additive PTAS when the valuation distribution is given explicitly, and an additive PRAS when the valuation distribution is given by a sampling oracle.
\end{theorem}
\begin{proof}
Lemma~\ref{lem:rev_stable} shows that the function $\gRev$ is 2-stable and 1-Lipschitz. 
If the valuation distribution $\D$ is explicitly given with support size $C$,
  the function $\gRev$ can be evaluated in $\poly(n,m,C)$ time.
Then for any $\epsilon > 0$,  it follows from Theorem~\ref{thm:signal} by 
  setting $\alpha=\eps/4$ and $\delta=\eps/2$ that there is a deterministic algorithm 
  that computes a signaling scheme with expected revenue $OPT-\eps$, in time
  $\poly(n,m^{\eps^{-2}\ln(1/\eps)},C)$.

If $\D$ is given via a sampling oracle,  standard tail bounds and the union bound imply that  $C= \Theta((s \log m + \log(\gamma^{-1})) / \eps^2)$ samples from $\D$ suffice to estimate to within $O(\epsilon)$ the revenue associated with every $s$-uniform posterior in $\Delta_m$, with success probability $1-\gamma$. Since revenue is $O(1)$-stable and $O(1)$-Lipschitz, Lemma \ref{lem:sparse_signals} implies that we can restrict attention to signaling schemes with $s$-uniform posteriors for $s=\poly(\frac{1}{\eps})$. Proceeding as in Theorem~\ref{thm:signal}, using the revenue estimates from Monte-Carlo sampling in lieu of exact values, we can construct a signaling scheme with revenue $OPT-\eps$  in time $\poly(n,m^{\eps^{-2}\ln(1/\eps)},\log(\frac{1}{\gamma}))$, with success probability $1-\gamma$.
%
%  similar to that in Theorem~\ref{thm:lottery_pras} shows that   $C= \Theta((s \log m + \log(\gamma^{-1})) / \eps^2)$ samples
%   we can establish the following:
% For any $\epsilon,\gamma > 0$, 
%    let $s =\Theta(\eps^{-2} \ln(1/\eps))$ and 
%   $r= \Theta((s \log m + \log(\gamma^{-1})) / \eps^2)$.
% There is a Monte Carlo algorithm that, with probability $1-\gamma$,
%   constructs a signaling scheme that has expected revenue $OPT-\eps$
%    in time $poly(n,m^{\eps^{-2}\ln(1/\eps)},r)$.
\end{proof}

%%% Local Variables: 
%%% mode: latex
%%% TeX-master: "main"
%%% End: 

%% file: voting.tex
%---- macros used by Ehsan in this section:
% \Set, \SpecSet, \Size
% \Ordinal
% \numrow, \numcol
% \gdom, \gindices
% \Threshold

% local macros for this section macros:

\def\VOT{Voting problem}
\def\pDist{p}
\def\signal{\sigma}
\def\state{\theta}
\def\State{\Theta}
\def\postFunc{\mu}
\def\postSet{P}
\def\postCo{\pi}

\section{Persuading Voters}\label{sec:voting}
In this section, we apply our mixture selection framework
  to natural signaling problems encountered in the context of social choice, as introduced by Alonso and C\^{a}mara \cite{Alonso14}.
Consider an election with two possible outcomes, `Yes' and `No'. 
For example,  voters may need to choose whether to adopt a new law or social policy;  board members of
  a company may need to decide whether to invest in a new project;
 and  members of a jury must decide whether a defendant is declared guilty or not guilty.
As in \cite{Alonso14}, we focus on the scenario in which voters have uncertainty regarding their utilities for the two possible outcomes (e.g., the risks and rewards of 
 the new project). Specifically, voters' utilities are parameterized by an a-priori unknown state of nature $\theta$ drawn from a common-knowledge prior distribution. We adopt the perspective of a principal  with access to the realization of $\theta$, and looking to influence the outcome of the election by signaling. 

Formally, we consider a voting setting with $\numrow$ voters and $\numcol$ states of nature. We index the voters by the set $[n]=\set{1,\ldots,n}$, and states of nature by the set $\Theta = \set{1,\ldots, m}$. We assume voters' preferences are given by a matrix $U \in [-1,1]^{n \times m}$, where $U_{i,j}$ denotes voter $i$'s utility in the event of a `Yes' outcome in state of nature $j$. Without loss of generality, we assume utilities are normalized so that each voter's utility for a `No' outcome is $0$ in each state of nature.  A voter $i$ who believes that the state of nature follows a distribution $\mu \in \Delta_m$ has expected utility $u(i,\mu) = \sum_{j \in \Theta} U_{i,j} \mu_{j}$ for a `Yes' outcome. In most  voting systems with a binary outcome, including for example threshold voting rules, it is a dominant strategy to vote  `Yes' if the utility $u(i,\mu)$ is at least $0$ and `No' otherwise. For our approximation algorithms, we also allow implementation in approximate dominant strategies --- i.e., we sometimes assume a voter votes `Yes' if his utility $u(i,\mu)$ is at least $-\delta$ for a small parameter $\delta$.\footnote{Such relaxations seem necessary for our results. Moreover, depending on the context, modes of intervention for shifting the votes of voters who are close to being indifferent may be realistic.} We assume that the state of nature $\theta \in \Theta$ is drawn from a common prior $\lambda \in \Delta_m$, and a principal with access to $\theta$ reveals a public signal $\sigma$ prior to voters casting their votes. As usual, we adopt the perspective of a principal looking to commit to a signaling scheme $\varphi: \Theta \to \Sigma$, for some set of signals $\Sigma$. 

Alonso and C\^{a}mara \cite{Alonso14} consider a principal  interested in maximizing the probability that at least $50\%$ (or some given threshold) of the voters vote 'Yes', in expectation over states of nature. They characterize optimal signaling schemes  analytically, though stop short of prescribing an algorithm for signaling. Theirs is the natural objective when the election employs a majority (or threshold) voting rule, and the principal is interested in influencing the outcome of the vote. Approximating this objective requires nontrivial modifications to our framework, and therefore we begin this section by examining a different, yet also natural, objective: the expected number of 'Yes' votes. We design a bi-criteria approximation scheme for this objective, then describe the necessary modifications for the threshold function objective of \cite{Alonso14}.

\subsection{Maximizing Expected Number of Votes}
We now examine bi-criteria approximation algorithms for maximizing the expected number of `Yes' votes. For our benchmark, we use the function  $\gVotesum(t) := \sum_{i \in [n]} \frac{1}{n} I[t_i \ge 0]$, where $I[\E]$ denotes the indicator function for event $\E$. Assuming voters vote `Yes' precisely when their posterior expected utility for a `Yes' outcome is nonnegative, the number of `Yes' votes when voters have preferences $U \in [-1,1]^{n \times m}$ and posterior belief $\mu \in \Delta_m$ equals $\gVotesum(U \mu)$. When the state of nature is distributed according to a common prior $\lambda$, and voters are informed according to signaling $\varphi$ inducing a decomposition $\set{\mu_\sigma, \nu_\sigma}_{\sigma \in \Sigma}$ of $\lambda$, the expected number of `Yes' votes equals $\FVotesum(\varphi, U,\lambda) := \sum_{\sigma \in \Sigma} \nu_\sigma \gVotesum(U \mu_\sigma)$. We use $\OVotesum(U,\lambda)$ to denote the maximum value of $\FVotesum(\varphi,U,\lambda)$ over signaling schemes $\varphi$.%, assuming states of nature are distributed according to $\lambda \in \Delta_m$ and preferences are given by $U \in [-1,1]^{n \times m}$.

As the first step to apply our framework, we prove that $\gVotesum$ is stable.

%While it remains open if an additive (single-criteria) PTAS signaling scheme exists for the 
%  voter persuasion problem, below we give a bi-criteria PRAS algorithm for this signal
%  problem.
%\yucomment{This sentence is no longer true. For the new voting function (AC14), we can show that if we are not allowed to lose $\epsilon$ in the voter's incentive (i.e. Lipschitz extension), even if we could lose $\epsilon$ in fraction-threshold and objective, the problem remains NP-Hard.}

\begin{lemma}\label{fac:voting:stable}
The function $\gVotesum$
is \gStable{1}.
\end{lemma}
\begin{proof}
For each voter $i \in \gindices$, let 
$g_{i} : \gdom \rightarrow \Set{0, 1}$ be the function
indicating whether voter $i$ prefers the `Yes' outcome,
i.e., $g_i(\ginp) = I[t_i \ge 0]$.
Each individual $g_i$ is \gStable{1},
  because as long as the $i$'th input $\ginp_i$
  is not corrupted the output of $g_i$ does not change.
Therefore $\gVotesum(\ginp) = 
\frac{1}{\numrow} \sum\limits_{i = 1}^{\numrow} g_i(\ginp)$, being a convex combination of  $1$-stable functions,
is \gStable{1} by Proposition~\ref{prop:close}.
\end{proof}

Unfortunately, $\gVotesum$ is not $O(1)$-Lipschitz. We therefore employ the bi-criteria extension to our framework from Definition~\ref{def:relaxation}. Specifically, for a parameter $\delta>0$, we assume a voter votes `Yes' as long as his expected utility from a `Yes' outcome is at least $-\delta$. Correspondingly, we  define the relaxed function $\gVotesum_\delta (t) := \sum_{i \in [n]} \frac{1}{n} I[t_i \ge -\delta]$; the expected number of `Yes' votes from a signaling scheme $\varphi=\set{\mu_\sigma,\nu_\sigma}_{\sigma \in \Sigma}$ can analogously be written as $\FVotesum_\delta(\varphi, U,\lambda) := \sum_{\sigma \in \Sigma} \nu_\sigma \gVotesum_\delta(U \mu_\sigma)$. 

It is easy to verify that  $\gVotesum_\delta$ is an $(\delta,0)$-relaxation of $\gVotesum$; combining this fact with Theorem \ref{thm:signal:bicriteria} yields a bi-criteria approximation scheme for the problem of maximizing the expected number of `Yes' votes.

\begin{theorem}
\label{thm:voting_sampling}
Let $\eps,\delta >0$ be parameters, let $U \in [-1,1]^{n \times m}$ describe the preferences of $n$ voters in $m$ states of nature, and let $\lambda \in \Delta_m$ be the prior of states of nature. There is an algorithm with runtime $\poly(m^{\frac{\ln(1/\eps)}{\delta^2}},n)$ for computing a signaling scheme $\varphi$ such that $\FVotesum_\delta(\varphi,U,\lambda) \geq \OVotesum(U,\lambda) - \eps$.
\end{theorem}

Using the same techniques as in Section \ref{sec:lottery:oracle}, 
  we can extend this result to the case where
  the valuations of voters are drawn from a distribution give either explicitly or by a sampling oracle. We omit the details.

\subsection{Maximizing Probability of a Majority Vote}

We now sketch the necessary modifications when the principal is interested in maximizing the probability of a `Yes' outcome, assuming a majority voting rule. We make two relaxations, which appear necessary for our framework: we assume a voter votes `Yes' as long as his expected utility from a `Yes' outcome is at least $-\delta$, and assume that the `Yes' outcome is attained when at least a $(0.5 - \delta)$ fraction of voters vote `Yes'. Our benchmark will be the maximum probability of a 'Yes' outcome in the absence of these two relaxations.

We define our benchmark using the function $\gVotethresh(t) = I[\gVotesum(t) \geq 0.5]$  which evaluates to $1$ if at least half of its $n$ inputs are nonnegative, and to $0$ otherwise. This function is not $O(1)$-stable, so we work with a more stringent benchmark which is. Specifically, for a parameter $\delta > 0$, we use the function $\gVotesmooth_\delta$ which is pointwise greater than or equal to $\gVotethresh$, defined as follows: 
\[
\gVotesmooth_\delta(t) = \begin{cases}  
\frac{1}{\delta} \left(\gVotesum(t) -0.5 + \delta \right)  &\mbox{if } \gVotesum(t) \in [0.5 - \delta, 0.5]  \\ 
\gVotethresh(t) & \mbox{otherwise.} \end{cases}  \]
Observe that $\gVotesmooth_\delta$ applies a continuous piecewise-linear function to the output of $\gVotesum$. It is easy to verify that $\gVotesmooth_\delta$ is $\frac{1}{\delta}$-stable, and upperbounds $\gVotethresh$.

Finally, to measure the quality of our output we define the relaxed function $\gVotethresh_\delta : [-1,1]^n \to \set{0,1}$, which outputs $1$ if at least a $(0.5 - \delta)$ fraction of its inputs exceed $-\delta$, and outputs $0$ otherwise.
By Definition~\ref{def:relaxation},
$\gVotethresh_\delta$ is a $(\delta,0)$-relaxation of $\gVotesmooth_\delta$ (and, consequently, also of $\gVotethresh$).

As usual, let $\FVotethresh(\varphi,U,\lambda)$ and $\FVotethresh_\delta(\varphi,U,\lambda)$ denote the functions which evaluate the quality of a signaling $\varphi$ scheme using $\gVotethresh$ and $\gVotethresh_\delta$, respectively. Moreover, let $\OVotethresh(U,\lambda)$ be the maximum value of $\FVotethresh(\varphi,U,\lambda)$ over signaling schemes $\varphi$. We apply Theorem \ref{thm:signal:bicriteria} to $\gVotethresh_\delta$ and $\gVotesmooth$, setting $\alpha=\eps \delta$, and use the fact that $\gVotesmooth$ upperbounds our true benchmark $\gVotethresh$, to conclude the following. 

\begin{theorem}
\label{thm:voting_sampling}
Let $\eps,\delta >0$ be parameters, let $U \in [-1,1]^{n \times m}$ describe the preferences of $n$ voters in $m$ states of nature, and let $\lambda \in \Delta_m$ be the prior of states of nature. There is an algorithm with runtime $\poly(m^{\frac{\ln(1/\eps\delta)}{\delta^2}},n)$ for computing a signaling scheme $\varphi$ such that $\FVotethresh_\delta(\varphi,U,\lambda) \geq \OVotethresh(U,\lambda) - \eps$.
\end{theorem}

\subsection{Connection to Maximum Feasible Subsystem of Linear Inequalities}
Turning our attention away from signaling, we note that $\gVotesum(Ax)$ simply counts the number of satisfied inequalities in the system $Ax \succeq 0$. Mixture selection for $\gVotesum$ is therefore the problem of maximizing the number of satisfied inequalities over the simplex. Using our framework from Section \ref{sec:framework}, we obtain a bi-criteria PTAS for this problem. Moreover, using Monte-Carlo sampling, our bi-criteria PTAS extends to the model in which $A$ is given implicitly; specifically, the rows of $A$ correspond to the sample space of a distribution $\D$ over $[-1,1]^m$, and are weighted accordingly. In this implicit model, we can think of mixture selection for $\gVotesum$ as the problem of finding  $x \in \Delta_m$ which maximizes the probability that $a \cdot x \geq 0$ for $a \sim \D$.

Motivated by systems applications, Daskalakis et al.~\cite{Daskalakis:2014} consider a special case of this problem termed {\em Fault-Tolerant Distributed Storage}.
Their problem is equivalent to mixture selection for $\gVotesum$ in the implicit model, with the additional restriction that $\D$ is a product distribution over binary vectors with marginal probabilities given explicitly.
They present an additive EPTAS for this problem in a uni-criteria sense.
Our framework relaxes their restrictions on $\D$, at the cost of a bi-criteria guarantee and exponential dependence on the error parameters.

%%% Local Variables: 
%%% mode: latex
%%% TeX-master: "main"
%%% End: 

%% file: game.tex
%% Signaling in Normal Form Games
\newcommand{\eq}{\textit{EQ}}

\section{Signaling in Bayesian Normal Form Games}\label{sec:game}

We consider normal form games of incomplete information,
  in which  payoffs are parameterized by a state of nature $\theta$.
A principal has access to the  exact realization of $\theta$, whereas the players initially share a prior belief on $\theta$ and form a posterior belief based on the information revealed by the principal.
The goal of the principal is then to commit to revealing certain information about $\theta$ ---
  i.e., a signaling scheme --- to induce a favorable equilibrium
  over the resulting Bayesian subgames.

Signaling in normal form games has recently been examined
  from a complexity-theoretic perspective.
Dughmi \cite{D14} considered the special case of two-player zero-sum games,
  and examined the design of symmetric signaling schemes
  with the goal of maximizing the expected utility of one of the players.
It was shown that no FPTAS
  is possible for the signaling problem for zero sum games,
  assuming the planted clique conjecture.
In this section, we complement the impossibility result of \cite{D14}
  with a  bi-criteria  quasi-polynomial time approximation scheme (QPTAS)
  which applies to %general-sum 
  normal form games with a constant number of players,
  slightly relaxing both the equilibrium definition 
  and the polynomial-time restriction.
It remains open if signaling for Bayesian zero-sum games admits a PTAS.

\subsection{Background and Notation}
We make heavy use of \emph{tensors} in describing multi-player games. Specifically, we focus on $n$-dimensional tensors of order $k$ with entries in $[-1,1]$, where $n$ is typically the number of strategies per player and $k$ is the number of players. We think of these tensors as functions $T: [n]^k \to [-1,1]$ mapping a strategy profile $s \in [n]^k$ to a  number in $[-1,1]$. Such a tensor $T$ also naturally describes a multilinear map; overloading notation, given $k$ vectors  $x_1,\ldots,x_k \in \RR^n$ we write $T(x_1\dots,x_k) = \sum_{s_1,\ldots,s_k \in [n]} \left( T(s_1,\ldots,s_k) \cdot \prod_{i=1}^k x_i(s_i) \right)$. This is most natural when $x_1,\ldots,x_k \in \Delta_n$ form a mixed strategy profile, in which case $T(x_1,\ldots,x_k)$ evaluates the expected value of $T$ over pure strategy profiles drawn from $(x_1,\ldots,x_k)$.

A Bayesian normal form game is defined by the following parameters:
\begin{itemize}
\item An integer $\numppl$ denoting the number of players,
  indexed by the set $[\numppl]=\set{1,\ldots,\numppl}$.
\item  An integer $\numact$ bounding the number of pure strategies 
  of each player. Without loss of generality, we assume each 
  player has exactly $\numact$ pure strategies, and index them by the set $[\numact]=\set{1,\ldots,\numact}$.
\item An integer $m$ denoting the number of \emph{states of nature}. We index states of nature by the set $\Theta=\set{1,\ldots,m}$, and use the variable $\theta$ to represent a state of nature.
\item A \emph{common prior distribution}
  $\lambda \in \Delta_{m}$ on states of nature.
\item A family of payoff tensors $\A_i^\theta: [\numact]^\numppl \to [-1,1]$, 
  one per player $i$ and state of nature $\theta$, 
  where $\A_i^\theta(s_1,\ldots,s_\numppl)$ is the payoff to 
  player $i$ when the state of nature is $\theta$ and 
  each player $j$ plays strategy $s_j$.
\end{itemize}

Note that a game of complete information is the special case with $m=1$ --- i.e.,
 the state of nature is fixed and known to all.
In a general Bayesian normal form game, absent any information about
  the state of nature beyond the prior $\lambda$,
  risk neutral players will behave as in the
  complete information game $\Ex_{\theta \sim \lambda} [\A^\theta]$. 
We consider signaling schemes which partially and symmetrically inform
  players by publicly announcing a signal $\sigma$, correlated with
  $\theta$; this induces a common posterior belief on the state of
  nature for each value of $\sigma$. 
When players' posterior belief over
  $\theta$ is given by $\mu \in \Delta_m$, we use $\A^\mu$ to denote the
  equivalent complete information game $\Ex_{\theta \sim \mu} [\A^\theta]$.  
As shorthand, we use $\A_i^\mu(x_1,\ldots,x_\numppl)$ to
  denote $\Ex [\A_i^\theta(s_1,\ldots,s_\numppl)]$ when $\theta \sim \mu \in
  \Delta_m$ and $s_i \sim x_i \in \Delta_\numact$.
In the event that the state
  of nature is $\theta$ and players play the pure strategy profile
  $s_1,\ldots,s_\numppl$, we refer to the tuple $(\theta,s_1,\ldots,s_\numppl)$ as
  the \emph{state of play}.
For our result, we assume that a Bayesian game $(\A,\lambda)$
  is represented explicitly as a vector $\lambda \in \Delta_m$
  and a list of tensors 
  $\{\A_i^\theta \in [-1,1]^{\numact^\numppl}: i \in [\numppl], \theta \in [m]\}$.

We adopt the approximate Nash equilibrium as our equilibrium concept.
There are two variants.
\begin{definition}
Let $\eps \geq 0$.
In a $\numppl$-player $\numact$-action normal form game with
  expected payoffs in $[-1,1]$ given by tensors $\A_1,\ldots,\A_\numppl$,
  a mixed strategy profile $x_1,\ldots,x_\numppl \in \Delta_\numact$
  is an $\eps$-Nash Equilibrium ($\eps$-\textit{NE})
  if \[\A_i(x_1,\ldots,x_\numppl) \geq \A_i(t_i, x_{-i}) - \eps \]
  for every player $i$ and alternative pure strategy $t_i \in [\numact]$.
\end{definition}

\begin{definition}
Let $\eps \geq 0$.
In a $\numppl$-player $\numact$-action normal form game with
  expected payoffs in $[-1,1]$ given by tensors $\A_1,\ldots,\A_\numppl$,
  a mixed strategy profile $x_1,\ldots,x_\numppl \in \Delta_\numact$
  is an $\eps$-well-supported Nash equilibrium ($\eps$-WSNE)
  if \[\A_i(s_i, x_{-i}) \geq \A_i(t_i, x_{-i}) - \eps\]
  for every player $i$, strategy $s_i$ in the support of $x_i$,
  and alternative pure strategy $t_i \in [\numact]$.
\end{definition}

Clearly, every $\eps$-WSNE is also an $\eps$-NE. When $\eps=0$, both correspond to the exact Nash Equilibrium.
Note that we omitted reference to the state of nature in the above definitions
  --- in a subgame corresponding to posterior beliefs $\mu \in \Delta_m$,
  we naturally use tensors $\A^\mu_1,\ldots\A^\mu_\numppl$ instead.

Fixing an equilibrium concept (NE, $\eps$-NE, or $\eps$-WSNE),
  a Bayesian game $(\A,\lambda)$,
  and a signaling scheme $\varphi: \Theta \to \Sigma$,
  an \emph{equilibrium selection rule} distinguishes
  an equilibrium strategy profile $(x_1^\sigma,\ldots,x^\sigma_\numppl)$
  to be played in each subgame $\sigma$ --- we call the tuple
  $X=\set{x^\sigma_i: \sigma \in \Sigma, i \in [\numppl]}$
  a \emph{Bayesian equilibrium} of the game $(\A,\lambda)$
  with signaling scheme $\varphi$.
Together with the prior $\lambda$, the Bayesian equilibrium $X$
  induces a distribution $\Gamma \in \Delta_{\Theta \cross [\numact]^\numppl}$
  over states of play --- we refer to $\Gamma$ as a \emph{distribution of play}.
% We say $\Gamma$ is \emph{implemented} by signaling scheme $\varphi$ in the chosen equilibrium concept (be it NE, $\eps$-NE, or $\eps$-WSNE).
This is analogous to implementation of allocation rules in traditional mechanism design.

Our results concern objectives which depend only on the state of play,
  and we seek to maximize the objective in expectation over the
  distribution of play.
These include, but are not restricted to, the social welfare of the players,
  as well as weighted combinations of player utilities.
Formally, our objective is described by a family of tensors $\FFF^\theta: [\numact]^\numppl \to [-1,1]$, one for each state of nature $\theta \in \Theta$. Equivalently, we may think of the objective as describing the payoffs of an additional player in the game --- namely the principal. For a distribution $\mu$ over states of nature, we use $\FFF^\mu = \Ex_{\theta \sim \mu} \FFF^\theta$ to denote the principal's expected utility in a subgame with posterior beliefs $\mu$, as a function of players' strategies.
 
For a signaling scheme $\varphi$ and associated (approximate) equilibria $X=\set{x^\sigma_i: \sigma \in \Sigma, i \in [\numppl]}$, our objective function can be written as $F(\varphi,X) = \Ex_{\theta \sim \lambda}
  \Ex_{\sigma \sim \varphi(\theta)} \Ex_{\vec{s} \sim x^\sigma} [ \F(\theta, \vec{s})]$. When $\varphi$ corresponds to a convex decomposition $\set{(\mu_\sigma, \nu_\sigma)}_{\sigma \in \Sigma}$ of the prior distribution, this can be equivalently written as $F(\varphi,X) = \sum_{\sigma \in \Sigma} \nu_\sigma \F^{\mu_\sigma} (x^\sigma)$.
%Naturally, when $\mu$ is a distribution over states of nature and $x$ is a mixed strategy profile, we denote $\FFF(\mu,x) = \Ex_{\theta \sim \mu} \Ex_{\vec{s} \sim x} \FFF(\theta,\vec{s})$.
%
 Let $OPT=OPT(\A,\lambda, \F)$ denote the maximizer of $F(\varphi^*,X^*)$ over signaling schemes $\varphi^*$ and (exact) Nash equilibria $X^*$. We seek  a signaling scheme $\varphi: \Theta \to \Sigma$,
  as well as a Bayesian $\eps$-NE (or $\eps$-WSNE) $X$ such that $F(\varphi,X) \geq OPT - \eps$. 

We will use the following Lemma, which follows easily from the results of Lipton et al.~\cite{lmm03}, to restrict attention to equilibria with small support.
%\yucomment{Do we bother to also cite a more recent result by Sid, which enables setting $r(\eps) = O(\frac{\log k + \log n}{\eps^2})$.}
\begin{lemma}\label{lem:lmm03}
  Let tensors $\A_1,\ldots,\A_k : [n]^k \to [-1,1]$ describe a $k$-player game of complete information with $n$ pure strategies per player, and let $\F:[n]^k \to [-1,1]$  be a tensor describing an objective function on mixed strategies. Define the function $r(\eps) =\frac{3 (k+1)^2 \ln ((k+1)^2 n)}{\eps^2}$. For each $\eps >0$, integer $s \geq r(\eps)$,  and mixed strategy profile $x=(x_1,\ldots,x_k)$, there is a profile $\tilde{x}=(\tilde{x}_1,\ldots,\tilde{x}_k)$ of $s$-uniform mixed strategies such that $|\A_i(x) - \A_i(\tilde{x})| \leq \eps$ for all players $i$, $|\F(x) - \F(\tilde{x})| \leq \eps$, and if $x$ is a Nash equilibrium of $\A$ then $\tilde{x}$ is an $\eps$-equilibrium of $\A$. This holds for both NE and WSNE.
\end{lemma}
\begin{proof}
We can think of the tensor $\F^\mu : [n]^k \to [-1,1]$ as describing the utility of an additional player in the game with a trivial strategy set. The rest follows from \cite[Theorem 2]{lmm03}.
\end{proof}

\subsection{QPTAS for Signaling in Normal Form Games}
We prove the following bi-criteria result.

\begin{theorem}
\label{thm:gamemain}
Let $\eps >0$ denote an approximation parameter,  let $(\A,\lambda)$ be a Bayesian normal form game with $\numppl=O(1)$ players, $\numact$ actions, and $m$ states of nature, and let $\F:[m] \cross [\numact]^\numppl \to [-1,1]$ be an objective function given as a tensor. 
  There is an algorithm with runtime $\poly(m^{\frac{\ln(n/\eps)}{\eps^2}}, \numact^{\frac{\ln \numact}{\eps^2}})$
  which outputs a signaling scheme $\varphi$
  and corresponding Bayesian $\eps$-equilibria $X$ satisfying
  $F(\varphi,X) \geq OPT(\A,\lambda, \F) - \eps$.
  This holds for both approximate NE and approximate WSNE.
\end{theorem}
In other words, when the number of players is a constant we can in quasi-polynomial time
  approximate the optimal reward from signaling while losing an additive
  $\eps$ in the objective as well as in the incentive constraints, as compared to the optimal signaling scheme / Nash equilibrium combination.

Fix $\eps > 0$. To prove this theorem, we define functions $g$ and $g_\eps$ which each take as input a $k$-player $n$-action game of complete information $\B$, given as payoff tensors $\B_1 \ldots, \B_k : [n]^k \to [-1,1]$, and an objective tensor $\G: [n]^k \to [-1,1]$, and output a number in $[-1,1]$. Specifically, $g(\B,\G) = \max\{\G(x) : x \in \eq(\B) \}$ and  $g_\eps( \B, \G) = \max\{\G(x) : x \in \eq_\eps(\B) \}$, where $\eq(\B)$ denotes the set of Nash equilibria of the game $\B$,  and $\eq_\eps(\B)$ denotes the (non-empty) set of $\ceil{r(\eps/4)}$-uniform $\eps$-Nash equilibria (or $\eps$-WSNE) for $r$ as given in Lemma \ref{lem:lmm03}. Recall that $\G(x)$ denotes evaluating the multilinear map described by tensor $\G$ at the mixed strategy profile $x \in \Delta_n^k$.

Now suppose we fix a Bayesian game $(\A,\lambda)$ and objective tensor $\F$ as in the statement of Theorem \ref{thm:gamemain}. For a subgame with a posterior distribution $\mu \in \Delta_m$ over states of nature, the principal's expected utility at the ``best'' Nash equilibrium of this subgame can be written as $g(\A^\mu, \F^\mu)$. Similarly, the principal's expected utility at the ``best'' $\ceil{r(\eps/4)}$-uniform $\eps$-NE (or $\eps$-WSNE) can be written as $g_\eps(\A^\mu, \F^\mu)$. Observe that the input to both $g$ and $g_\eps$ is a linear function of $\mu$, as need to apply the results in Section \ref{sec:signaling}. For a signaling scheme $\varphi$ corresponding to a decomposition   $\lambda = \sum_{\sigma \in \Sigma} \nu_\sigma \cdot \mu_\sigma$ of the prior distribution $\lambda$ into posterior distributions (see Section~\ref{prelim:schemes}), we can write the principal's expected utility assuming the ``best'' Nash equilibrium as $F(\varphi) = \sum_{\sigma \in \Sigma} \nu_\sigma g(F^{\mu_\sigma}, \A^{\mu_\sigma})$, and assuming the ``best'' $\ceil{r(\eps/4)}$-uniform $\eps$-equilibrium as $F_\eps(\varphi) = \sum_{\sigma \in \Sigma} \nu_\sigma g_\eps(F^{\mu_\sigma}, \A^{\mu_\sigma})$. We use $OPT$ to denote the maximum value of $F$ over all signaling schemes.

We prove Theorem \ref{thm:gamemain}  by exhibiting an algorithm for
  computing a signaling scheme $\varphi$ such that
  $F_{\eps}(\varphi)
   \geq OPT - \eps$.
The proof hinges on two main lemmas.

\begin{lemma}
\label{lem:gameStable}
The function $g$ is $2 (k+1) n^k$-stable.
\end{lemma}
\begin{proof}
As noted in Section \ref{sec:framework},
  any function mapping  a hypercube $[-1,1]^N$ to the interval $[-1,1]$ is $2N$ stable. The function $g$ is such a function with $N=  (k+1) n^k$. 
\end{proof}

\begin{lemma}
\label{lem:gameBicriteria}
The function $g_\eps$ is an $(\eps/4,\eps/2)$-relaxation of $g$.
\end{lemma}
\begin{proof}
Consider tensors $\G, \tilde{\G} : [n]^k \to [-1,1]$ with $|\G(s) - \tilde{\G}(s)| \leq \eps/4$ for all $s \in [n]^k$, and two $k$-player $n$-action games $\B=(\B_1,\ldots,\B_k)$ and $\tilde{\B}=(\tilde{\B}_1,\ldots,\tilde{\B}_k)$ with $|\B_i(s) - \tilde{\B}_i(s) | \leq \eps/4$ for all $s \in [n]^k$. It suffices to show that $g_\eps(\tilde{\B},\tilde{\G}) \geq g(\B,\G) - \eps/2$. Let $x \in \Delta_n^k$ be the Bayesian equilibrium of $\B$ for which $\G(x) = g(\B,\G)$.  
By Lemma \ref{lem:lmm03}, there is a profile $\tilde{x}$ of $\ceil{r(\eps/4)}$-uniform mixed strategies such that $\tilde{x}$ is an $\eps/4$-equilibrium of $\B$, and $\G(\tilde{x})  \geq \G(x) - \eps/4$. Since $\tilde{\B}$ differs from $\B$ by at most $\eps/4$ everywhere, it follows that $\tilde{x}$ is an $\eps$-equilibrium of $\tilde{\B}$, i.e. $\tilde{x} \in \eq_\eps(\tilde{\B})$. Similarly, since $\tilde{\G}$ differs from $\G$ by at most $\eps/4$ everywhere, it follows that $\tilde{\G}(\tilde{x}) \geq \G(\tilde{x}) - \eps/4 \geq \G(x) - \eps/2$. We conclude that $g_\eps(\tilde{B},\tilde{\G}) \geq \tilde{\G}(\tilde{x}) \geq g(\B,\G) - \eps/2$.
\end{proof}

We now complete the proof of Theorem \ref{thm:gamemain} by instantiating Theorem \ref{thm:signal:bicriteria} with $g$, $h=g_\eps$, and $\alpha = \frac{\eps}{4 (k+1)n^k}$. The runtime is $\poly(m^{\frac{\ln(1/\alpha)}{\eps^2}}, (k+1)n^k, T)$, where $T$ is the time needed to evaluate $g_\eps$ (and compute the corresponding $\ceil{r(\eps/4)}$-uniform $\eps$-equilibrium) on a given input. Recall that $k=O(1)$  and $\alpha = \frac{\eps}{\poly(n)}$. Moreover, using brute-force enumeration of all $\ceil{r(\eps/4)}$-uniform mixed strategy profiles we conclude that  $T$ is bounded by a polynomial in $n^{\frac{\ln n}{\eps^2}}$. Therefore our total runtime is $\poly(m^{\frac{\ln(n/\eps)}{\eps^2}},n^{\frac{\ln n}{\eps^2}})$, as needed.

\paragraph{Remarks} In the special case of two-player zero-sum games and a principal interested in maximizing one player's utility, as studied in \cite{D14}, our techniques lead to a more efficient approximation scheme and a uni-criteria guarantee. This is because the principal's payoff tensor $\G$ equals the payoff tensor $\B$ of one of the players (say, player 1), and consequently the function $g(\B,\G) = g(\B, \B) = \max_x \min_y x^\T \B y$ is $n^2$-stable and \Lip{2}. Its Lipschitz continuity follows from the fact that an $\eps$-equilibrium of a zero-sum game leads to utilities within $\eps$ of the equilibrium utilities. Moreover, evaluating $g$ now takes time $T=\poly(m,n)$. Theorem~\ref{thm:signal} instantiated with $\alpha = \frac{\eps}{4n^2}$ and $\delta = \eps/4$, leads to an algorithm with runtime $\poly(m^{\frac{\ln(n/\eps)}{\eps^2}},n)$, which outputs a signaling scheme $\varphi$ and corresponding Bayesian (exact) Nash-equilibria $X$ satisfying $F(\varphi,X) \geq OPT(\A,\lambda, \F) - \eps$.

%%% Local Variables: 
%%% mode: latex
%%% TeX-master: "main"
%%% End: 

%% file: hardness.tex
%---- macros used by Ehsan in this section:
% \Size, \SpecSet. \Ordinal
% \numrow, \numcol, \gindices, \Threshold

% local macros for this section macros:

\newcommand{\gThresh}{\ensuremath{g^{\text{(threshold)}}}\xspace}

\def\OPTIS{\ensuremath{\text{OPT}_\text{IS}}\xspace}

\section{Hardness Results}\label{sec:hardness}
In this section, we present hardness results which justify our assumptions, and exhibit the limitations of our techniques. Specifically, we show in Sections \ref{sec:hardness:nolip} and \ref{sec:hardness:nostable} that neither stability nor Lipschitz continuity alone suffices for an additive PTAS. In Section \ref{sec:hardness:lottery}, we show that even in the presence of Lipschitz continuity and noise stability, obtaining an additive FPTAS would imply P = NP.

\subsection{NP-Hardness in the Absence of Lipschitz Continuity}
\label{sec:hardness:nolip}

We now show that stability alone does not suffice for an additive PTAS for mixture selection, in general. First, we show that mixture selection for the $1$-stable function $\gVotesum$, presented in Section \ref{sec:voting}, does not admit a (uni-criteria) additive PTAS unless P = NP. Being that $\gVotesum$ is not continuous in any metric, we drive the point home by exhibiting a ``smoothed'' function $\gSlope$ which is \gStable{1} and \Lip{O(1)} with respect to  $L^1$, but not \Lip{O(1)} with respect to $L^\infty$, and show that  \prob for $\gSlope$
still does not admit an additive PTAS unless P = NP. 

Both NP-hardness results share a similar reduction
from the \emph{maximum independent set} problem.
We use a consequence of the result by \cite{IS:hardness}, namely that there exists a constant $\epsilon$ such that it is NP-hard to approximate maximum independent set to within an additive error of $\epsilon n$, where $n$ denotes the number of vertices. 

Given an $n$-node undirected graph $G$,
let $\OPTIS=\OPTIS(G)$ be the size of its largest independent set.
We define the $n \times n$ matrix $A=A(G)$ as follows:
\begin{itemize}
\item Diagonal entries of $A$ are all $\frac{1}{2}$
($A_{i, i} = \frac{1}{2}$ for all $1 \leq i \leq n$).
\item When vertices $i$ and $j$ share an edge in $G$,
both $A_{i, j}$ and $A_{j, i}$ are $-1$.
\item All other entries of $A$,
namely $A_{i, j}$ for non-adjacent distinct vertices $i$ and $j$,
are $-\frac{1}{4 n}$.
\end{itemize}
We relate $\OPTIS$ to convex combinations of the columns of $A$ as follows.

\begin{observation}\label{obs:fromIS}
Let \III be an independent set of $G$ with $\Size{\III} = k$.
There exists $x \in \Delta_n$ such that $k$ entries of $A x$ are at least  $\frac{1}{4n}$, and all remaining entries are strictly negative.
\end{observation}
\begin{proof}
Let $x \in \Delta_n$ be the normalized indicator vector of $\I$ --- i.e., $x_i = \frac{1}{k}$ if $i \in \III$
and $x_i = 0$ otherwise. By construction $(A x)_i= \frac{1}{k} (\frac{1}{2} - (k-1) \frac{1}{4n}) \geq \frac{1}{4n}$ whenever $i \in \I$, and $(A x)_i \leq -\frac{1}{4n}$ otherwise.
\end{proof}

\begin{observation}\label{obs:toIS}
For any $x \in \Delta_n$, nonnegative entries of $A x$ correspond to an independent set of $G$.
Consequently, $A x$ can have at most \OPTIS nonnegative entries.
\end{observation}
\begin{proof}
Let $t = A x$.
Consider an edge $\Set{i,j}$ of graph $G$, and without loss of generality assume that $x_i \ge x_j$. If $x_i=0$, then $t_i \leq - \frac{1}{4n} < 0$ by construction. Othewise, $t_j \leq  \frac{x_j}{2} - x_i < 0$. 
Therefore,  $t_i$ and $t_j$ cannot be both nonnegative.
We conclude that the nonnegative coordinates of $t$ correspond to an independent set of $G$.
\end{proof}

Observations~\ref{obs:fromIS} and \ref{obs:toIS} imply that $\max_{x \in \Delta_n} \gVotesum(A x) = \frac{\OPTIS}{n}$. Combined with the fact that obtaining an additive PTAS for the maximum independent set problem is NP-hard, we get the following theorem.
\begin{theorem}\label{thm:voting:hard}
Mixture selection for the $1$-stable function $\gVotesum$ admits no additive PTAS unless P = NP. 
\end{theorem}

Noting that $\gVotesum$ is a discontinuous function, for emphasis we exhibit a function $\gSlope$ which is Lipschitz continuous in $L^1$ (but not in $L^\infty$) and $1$-noise stable, but for which the same impossibility result holds by an identical reduction. Informally, $\gSlope$ ``smoothes'' the threshold behavior of $\gVotesum$ as follows: each input $t_i$ contributes $0$ to $\gSlope(t)$ when $t_i \leq 0$, contributes $\frac{1}{n}$ when $t_i \geq \frac{1}{4n}$, and the contribution is a linear function of $t_i$ increasing from $0$ to $\frac{1}{n}$ for $t_i \in [0,\frac{1}{4n}]$. Formally, we define $\gSlope(t) = \sum_{i=1}^n \min\set{4 \max\set{0,t_i}, \frac{1}{n}}$. Since each entry of $t$ contributes at most $\frac{1}{n}$ to $\gSlope(t)$, it is easy to verify that $\gSlope$ is $1$-stable. Moreover, since the partial derivatives of $\gSlope(t)$ are upper-bounded by $4$, it is $4$-Lipschitz continuous with respect to the  $L^1$ metric.  Observations~\ref{obs:fromIS} and \ref{obs:toIS} imply that $\max_{x \in \Delta_n} \gSlope(A x) = \frac{\OPTIS}{n}$, ruling out an additive PTAS for mixture selection for $\gSlope$. % $\gSlope(t) = 4 \sum\limits_{i = 1}^{n} \min\Set{\max\Set{t_i, 0}, \frac{1}{4n}}$ 

\begin{theorem}\label{thm:gSlope:hard}
The function $\gSlope$ is $1$-stable and $O(1)$-Lipschitz with respect to $L^1$, and yet mixture selection for $\gSlope$ admits no additive PTAS unless P = NP. 
\end{theorem}

\subsection{Planted-Clique Hardness in the Absence of Stability}
\label{sec:hardness:nostable}

We now present evidence that Lipschitz continuity alone does not suffice for a PTAS for mixture selection. Recalling that a quasipolynomial time algorithm follows from our framework whenever a function is $O(1)$-Lipschitz, we reduce from the \emph{planted clique problem}---for which a quasipolynomial time algorithm exists, and yet a polynomial-time algorithm is conjectured not to exist---rather than from an NP-hard problem. 

In the planted clique problem, one must distinguish the $n$-node Erd\"os-R\'enyi random graph $\G(n,\frac{1}{2})$, in which each edge is included independently with probability $\frac{1}{2}$, from the graph $\G(n,\frac{1}{2},k)$ formed by ``planting'' a clique in $\G(n,\frac{1}{2})$ at a randomly (or, equivalently, adversarially) chosen set of $k$ nodes. This problem was first considered by by Jerrum~\cite{jerrum92} and Ku\u{c}era~\cite{kucera95}, and has been the subject of a large body of work since. A quasi-polynomial time algorithm exists when $k \geq 2 \log n$, and the best polynomial-time algorithms only succeed for $k = \Omega(\sqrt{n})$  (see e.g., \cite{alonclique} \cite{dekel} \cite{feigeron} \cite{coja}). Several papers suggest that the problem is hard for $k=o(\sqrt{n})$ by ruling out natural classes of algorithmic approaches (e.g. \cite{jerrum92, feigeprobable, feldmanclique}). The planted clique problem has therefore found use as a hardness assumption in a variety of applications (e.g. \cite{alontesting}, \cite{juels}, \cite{hazankrauthgamer}, \cite{minder}, \cite{D14}).
%
%
% The problem of recovering a planted $k$-clique
% from an Erd\"os-R\'enyi distribution $\GGG(n, \Half)$
% has been studied intensively since it was introduced
% by Jerrum~\cite{jerrum92} and Ku\u{c}era~\cite{kucera95}.
% Almost all of these works has focused on the case $p = \Half$.
% On the positive side, there is a quasipolynomial-time algorithm
% for recovering the clique when $k \geq 2 \log n$.
% The best known polynomial-time algorithms, on the other hand,
% can only recover planted cliques of size $\Omega(\sqrt{n})$
% \cite{alonclique} \cite{dekel} \cite{feigeron} \cite{coja}.
% Despite this extensive body of work,
% there is no known polynomial-time algorithm 
% for detecting planted cliques of size $k = o(\sqrt{n})$.
% There is evidence the problem is hard when $k = o(\sqrt{n})$:
% Jerrum \cite{jerrum92} ruled out Markov chain approach,
% Feige and Krauthgamer~\cite{feigeprobable} ruled out algorithms
% based on the Lov{\'a}sz--Schrijver family
% of semi-definite programming relaxations;
% and recently Feldman et al.~\cite{feldmanclique}
% defined and
% ruled out a family of ``statistical algorithms.''
% Consequently, the planted clique conjecture has been used as
% a hardness assumption in a variety of different applications
% (e.g. \cite{alontesting}, \cite{juels}, \cite{hazankrauthgamer}, \cite{minder}).
%
%Let $\GGG(n, \Half, k)$ denote the distribution over graphs
%generated by adding a $k$-clique into a graph drawn from $\GGG(n, \Half)$.
We use the following well-believed conjecture as our hardness assumption.
\begin{assumption}
\label{assumption:clique}
For some function $k=k(n)$ satisfying $k=\omega(\log^2 n)$ and $k=o(\sqrt{n})$,
%There exists a constant $0 < \gamma < \Half$, such that for $k = n^\gamma$,
  there is no probabilistic polynomial-time algorithm that can 
  distinguish between a random graph drawn from $\GGG(n, \Half)$
  and a random graph drawn from $\GGG(n, \Half, k)$ with success probability $1 - o(1)$.
\end{assumption}

We let $k=k(n)$ be as in Assumption \ref{assumption:clique}, and consider mixture selection for the function $\gclique_k: [0,1]^n \to [0,1]$ with $\gclique_k(\ginp) = \ginp_{[k]} - \ginp_{[k+1]} + \ginp_{[n]}$, where $t_{[i]}$ denotes the $i$'th largest entry of the vector $t$. It is easy to verify that  $\gclique_k$ is 3-Lipschitz with respect to the $L^\infty$ metric, yet is not $O(1)$-stable. We prove the following theorem.
\begin{theorem}\label{thm:hardness:nostable}
   Assumption \ref{assumption:clique} implies that there is no additive PTAS for mixture selection for $\gclique_k$.
\end{theorem}
\noindent To prove Theorem \ref{thm:hardness:nostable}, we show that $\max_{x \in \Delta_n} \gclique_k(A x)$ is arbitrarily close to $1$ with high probability when $A$ is the adjacency matrix of $G \sim \G(n,\frac{1}{2},k)$, and is bounded away from $1$ with high probability when $A$ is the adjacency matrix of $G \sim \G(n,\frac{1}{2})$. For convenience, and without loss of generality, we assume that both random graphs include each self-loop with probability $\frac{1}{2}$ --- i.e., diagonal entries of the adjacency matrix $A$ are independent uniform draws from $\set{0,1}$ in both cases. Our argument is captured by the following two lemmas.

% In this subsection, we will always set $\gamma$ to be the constant in
%   Assumption~\ref{assumption:clique} and $k = n ^ \gamma$.

% We consider the \Lip{3} function
%   $\gclique_k(\ginp) = \ginp_{[k]} - \ginp_{[k+1]} + \ginp_{[n]}$.
% Recall that in our notation,
%   $\ginp_{[i]}$ is the \Ordinal{i} largest entry of $\ginp$.
% Note that $\gclique_k$ is not \gStable{O(1)} in general.
% We exhibit a reduction from
%   the planted $k$-clique problem to \prob for $\gclique_k$,
%   and show that an $\eps$-approximation of \prob for $\gclique_k$
%   would reveal the planted clique in $\GGG(n, \Half, k)$.
% Therefore, a PTAS for \prob for $\gclique_k$ would
%   refute Assumption~\ref{assumption:clique}.

% Let $G$ be a random graph drawn from either
% $\GGG(n, \Half)$ or $\GGG(n, \Half, k)$.
% Let $A$ be its adjacency matrix.
% Lemmas~\ref{lem:planted} and \ref{lem:randomgraph}
% imply that if a polynomial-time algorithm
% can optimize \prob for $\gclique_k$
% within an additive error $\frac{1}{4} - 2\epsilon$ for any constant $\epsilon$,
% we can distinguish draws from $\GGG(n, \Half)$ and $\GGG(n, \Half, k)$
%  in polynomial time, with probability $1 - o(1)$. 

% We first show that if $G$
% is drawn from $\GGG(n, \Half, k)$,
% maximum of $\gclique_k(A x)$ is close to $1$
% with high probability.

\begin{lemma}\label{lem:planted}
Fix a constant $\epsilon > 0$. Let $G \sim \GGG(n, \Half, k)$, and let $A$ be its adjacency matrix. With probability $1-o(1)$, there
 exists an $x \in \Delta_n$ such that $\gclique_k(A x) \geq 1 - \epsilon$.
\end{lemma}

\begin{proof}
Let $\CCC$ denote the vertices of the planted $k$-clique.
We set $x_i = \frac{1}{k}$ if $i \in \CCC$ and $0$ otherwise.
Let $t = A x$.
For $i \in \CCC$, $t_i \geq 1 - \frac{1}{k}$.
On the other hand,
all other entries of $t$ concentrate around $\Half$
with high probability.
For $i \notin \CCC$, $t_i$ is simply the average of $k$ independent Bernoulli random variables
by definition of $\GGG(n, \Half, k)$; using Hoeffding's inequality, we bound the probability that $t_i$ deviates from its expectation by more than a constant $\delta>0$, to be chosen later:
\[
  \Pr \left[\left|t_i - \Half \right| > \delta \right]
  \leq 2 e^{-2 \delta^2 k}.
\]

By the union bound,
 $t_i \in [\Half - \delta, \Half + \delta]$
simultaneously for all $i \notin \CCC$
with probability at least $1 -  2^{\log n - \Omega(k)} = 1-o(1)$.
Thus $t_{[k+1]} - t_{[n]} \leq 2 \delta$ and
$\gclique_k(t) = t_{[k]} - (t_{[k+1]} - t_{[n]})
\geq 1 - \frac{1}{k} - 2 \delta$ with probability $1-o(1)$. Choosing $\delta = \eps/3$, we conclude that  $\gclique_k(t) \ge 1 - \eps$ with probability $1 - o(1)$.
\end{proof}

% We now show the other side of the reduction:
% when the graph $G$ is drawn from $\GGG(n, \Half)$,
% the maximum value of $\gclique_k(A x)$
% is at most $\frac{3}{4} + \epsilon$
% with high probability.
% Intuitively, for a fixed $x \in \Delta_n$,
% entries of $t = Ax$ are independent draws from a symmetric distribution,
% the probability of drawing a sample $t_i$ close to $1$ is the same as drawing $t_i$ close to $0$.
% But in order to maximize $\gclique_k(t)$,
% we need $k$ large samples, and $n-k$ samples which are close to each other.
% To this end, we bound the probability that
% $\gclique_k(A x) > \frac{3}{4} + \epsilon$ for a fixed $x$.
% We then exploit the fact that
% it suffices to union bound over quasipolynomial number of vectors in $\Delta_n$.

\begin{lemma}\label{lem:randomgraph}
Fix a constant $\epsilon > 0$. Let $G \sim \GGG(n, \Half)$, and let $A$ be its adjacency matrix. With probability $1-o(1)$,  
$\gclique_k(Ax) \leq \frac{3}{4} + \epsilon$ for all $x \in \Delta_n$.
\end{lemma}
\begin{proof}
Recall that $\gclique_k$ is $O(1)$-Lipschitz and --- like any other function from the hypercube to the bounded interval --- $O(n)$-stable. If there exists $x^*$ such that $\gclique_k(Ax^*) \geq \frac{3}{4} + \eps$, then Theorem \ref{thm:main}  implies that there is an integer $s=O(\log n)$ and an $s$-uniform vector $\tilde{x}$ such that $\gclique_k(A\tilde{x}) > \frac{3}{4}$. There are $n^s$ such vectors. We next show that for an arbitrary fixed vector $x \in \Delta_n$ the probability that $\gclique_k(A x) > \frac{3}{4}$  is at most $2^{-\Omega(k)}$.  This would complete the proof by the union bound, since $1- n^s \cdot 2^{-\Omega(k)} = 1- o(1)$.

Fix $x \in \Delta_n$, and let $t=Ax$. Define $\D$ as the distribution supported on $[0,1]$ which is sampled as follows: draw $a$ uniformly from $\set{0,1}^n$, and output $a \cdot x$. Since $A$ is the adjacency matrix of $G \sim \G(n,\frac{1}{2})$, each entry $t_i$ of $t$ can be viewed as an independent draw from $\D$. We exploit a key property of $\D$ in our proof, namely the fact that $\D$ is \emph{symmetric} about $\frac{1}{2}$. Formally we mean that $\Pr_\D [ r] = \Pr_\D [1-r]$ for all $r \in [0,1]$, and this follows easily from the definition of $\D$.

Symmetry of $\D$ implies that $\Pr_{r \sim \D} [ r \geq \frac{1}{2}] = \Pr_{r \sim \D} [ r \leq \frac{1}{2}] \geq \frac{1}{2}$. Recalling that $k = o(n)$ and that entries of $t$ are independent draws from $\D$, the Chernoff bound implies that the following holds with probability  at least $1 - 2^{-\Omega(n)}$:
\begin{equation}\label{equ:range}
t_{[n]} \leq \Half \leq t_{[k + 1]}.
\end{equation}
If $\gclique_k(t) > \frac{3}{4}$,
then the following two conditions must hold:
\begin{enumerate}
\item\label{item:topk} $t_{[k]} > \frac{3}{4}$, and
\item\label{item:rest}  $t_{[k + 1]} - t_{[n]} < \frac{1}{4}$.
\end{enumerate}
Condition~\ref{item:topk}
implies that the $k$ largest entries of $t$
are all at least $\frac{3}{4}$.
Furthermore, unless 
Inequality~\eqref{equ:range} is violated --- an event with small probability $2^{-\Omega(n)}$ ---
Condition \ref{item:rest} implies that remaining entries of $t$
are all strictly between $\frac{1}{4}$ and $\frac{3}{4}$. Let $p$ denote  $\Pr_{ r \sim \D} [ r  \leq \frac{1}{4}]$, also equal to  $\Pr_{r \sim \D} [ r \geq \frac{3}{4} ]$ by symmetry of $\D$. The probability that $k$ entries of $t$ are at least $\frac{3}{4}$ and all remaining entries are in $(\frac{1}{4},\frac{3}{4})$ is given by $\binom{n}{k} p^k (1-2p)^{n-k}$, which is maximized at $p = \frac{k}{2n}$, with maximum value $2^{-\Omega(k)}$. In summary, the probability that $\gclique_k(Ax) > \frac{3}{4}$ is at most $2^{-\Omega(k)} + 2^{-\Omega(n)} = 2^{-\Omega(k)}$, as needed.
\end{proof}

\subsection{NP-hardness of an additive FPTAS}
\label{sec:hardness:lottery}

Our last hardness proof rules out an additive FPTAS for the lottery design problem, a.k.a. mixture selection for the function $\gLot{}$, as defined in Section \ref{sec:lottery}.
We restrict attention to the weight vector $w$ assigning equal weight to all inputs to $\gLot{}$, and therefore omit references to the weight vector for the remainder of this section. %and rule out an additive fully polynomial time approximation scheme (FPTAS).
%and show that \prob for $\gLot{w}$
%is hard to approximate within $\Omega(\frac{1}{n^2})$ additive error.
%(We omit reference to the weight vector $w$ for the remainder of this section).

\begin{theorem}\label{thm:lottery:hard}
The lottery design problem, a.k.a. mixture selection for $\gLot{}$, admits no additive FPTAS unless P = NP.
\end{theorem}
\begin{proof}
The proof involves a reduction
from the independent set problem which is very similar to the reduction in Section \ref{sec:hardness:nolip}, so we only detail the necessary modifications. Given an $n$-node undirected graph $G$, we define an $n \times n$ matrix $A=A(G)$ as in Section \ref{sec:hardness:nolip}, though shifted and normalized so entries lie in $[0,1]$. Specifically, we set diagonal entries of $A$  to $\frac{3}{4}$, and we set an off-diagonal entry $A_{ij}$ to $0$ if $i$ and $j$ share an edge and to $\frac{1}{2} - \frac{1}{8n}$ otherwise. Observation \ref{obs:toIS}  implies that for every $x \in \Delta_n$, entries of $A x$ which are at least $\frac{1}{2}$ correspond to an independent set of $G$. Observation \ref{obs:fromIS} implies that there is a vector $x^* \in \Delta_n$ so that $A x^*$ has exactly $\OPTIS(G)$ entries no less than $p^* := \frac{1}{2} + \frac{1}{8n}$. 

Our input to the lottery design problem will be a valuation matrix $B$ with $8n^2 + n$ rows and $n$ columns, obtained from $A$ by adding $8n^2$ ``dummy'' rows, each of which is $(p^*,\ldots,p^*)$. Setting a price of $p^*$ for the lottery $x^*$ results in expected revenue $r^* := p^* \cdot \frac{8n^2 + \OPTIS}{8n^2 + n} = (\frac{1}{2} + \frac{1}{8n}) \cdot \frac{8n^2 + \OPTIS}{8n^2 + n} > \frac{1}{2}$. Therefore, $r^*$ lower-bounds $\max_{x \in \Delta_n} \gLot{}(B x)$. We claim that $r^*$ is in fact the maximum value of this mixture selection problem, and prove this by fixing an arbitrary lottery $x \in \Delta_n$ and conducting a simple case analysis on the associated price $p$:

\begin{itemize}
\item When $p^* < p \leq 1$, none of the ``dummy'' types purchase the lottery $x$, resulting in expected revenue at most $\frac{n}{8n^2 + n} < \frac{1}{2} < r^*$.
\item When $\frac{1}{2} \leq p \leq p^*$, all dummy types purchase $x$, and the $i$'th non-dummy type purchases $x$ only if $(A x)_i \geq \frac{1}{2}$. Since entries of $A x$ which are no less than $\frac{1}{2}$ correspond to an independent set of $G$, at most $8n^2 + \OPTIS$ types purchase the lottery $x$, resulting in expected revenue at most $p \cdot \frac{8n^2 + \OPTIS}{8n^2 + n} \leq r^*$.
\item When $0 \leq p < \frac{1}{2}$, the best case scenario is that all buyer types purchase the lottery at price $p$, yielding expected revenue at most $p < \frac{1}{2} < r^*$.
\end{itemize}

Recalling that there exists a constant $\eps > 0$ for which the maximum independent set problem does not admit an additive $\eps n$-approximation algorithm in polynomial time, we conclude that mixture selection for $\gLot{}$ does not admit a polynomial-time additive $\eps'$-approximation algorithm for any $\eps'=\eps'(n) \leq (\frac{1}{2} + \frac{1}{8n}) \frac{\eps}{8n +1}$. This rules out an additive FPTAS.
\end{proof}

%% file: fourier.tex
\newcommand{\HC}{\ensuremath{\{-1, 1\}^n}\xspace}
\newcommand{\Fstab}{Algebraic Stability\xspace}
\newcommand{\fstab}{algebraically stable\xspace}
\newcommand{\pD}{\ensuremath{p_{\DDD}}\xspace}
\newcommand{\Four}[1]{\widehat{#1}\xspace}
\newcommand{\Bone}{\bm{1}}
\newcommand{\kStab}{k}

\section{Connection to Boolean Function Analysis}
\label{sec:fourier}
In Definition \ref{def:strong-stability}, stability of a function from the solid hypercube to the bounded interval was defined as robustness to corruption patterns which follow a light distribution. In this section, we exhibit a closely-related algebraic notion of  stability, based on Fourier analysis of Boolean functions. Instead of using light distributions directly, we describe the effects of corruption on a function $g$ at input $t$ in the solid hypercube by a Boolean function $h_t$. In particular, $h_t$ takes as input a vector $z \in \HC$, interprets $z_i=1$ as forbidding corruption of input $t_i$ and $z_i=-1$ as permitting corruption of $t_i$, and considers the worst-case such corruption of $t$. Formally, denoting $\III^+(z) := \{i \in [n]: z_i = 1\}$ and $\III^-(z) := \{i \in [n] : z_i = -1\}$, we define Boolean extensions of $g$ at an input $t$.

% In the previous definition of stability, the inputs to $g$ (i.e., $t_1, \ldots ,t_n$) are considered to be corrupted by some random process that respects certain marginal constraints. Instead of using $\alpha$-light distributions, we can also capture such corruption by a Boolean indication vector $z \in \HC$ where $z_i = 1$ indicates that $t_i$ is \emph{not} corrupted while $z_i = -1$ indicates that $t_i$ can be arbitrarily corrupted.

% Formally, let $z \in \HC$ be a point in the hypercube. Let $\III^+(z) := \{i \in [n]: z_i = 1\}$ and $\III^-(z) := \{i \in [n] : z_i = -1\}$ ($\III^+(z)$ and $\III^-(z)$ are subsets of $[n]$ and form a partition of $[n]$). One can see that $\III^-(z)$ is the set of corrupted inputs.
% Given $t$ and $z$, we are interested in the worst corruption of inputs in $\III^-(z)$. The following definition encodes such information into a Boolean function.

\begin{definition}[Boolean extension]
\label{def:bool_ext}
 Let $g: [-1,1]^n \to [-1,1]$, and let $t \in [-1, 1]^n$.
A Boolean function $h_t:\HC \to [-1,1]$ is called a Boolean extension of $g$ at $t$ if it satisfies:
\begin{align}
h_t(\bm{1}) &= g(t) \\
h_t(z) &\le \min \{ g(t'): t' \corr{\I^-(z)} t \}.
\end{align}
\end{definition}
%
%So given function $g$ and point $t$ (recall $t = A \cdot y$ where $y \in \Delta_m$ and $m$ is number of `state-of-nature', so $t$ is essentially the posterior belief of everything), we defined a \emph{Boolean function} $h_t(x)$ according to definition \ref{def:bext}.
Next we define \Fstab for $g(t)$, using the notion of the Fourier transform of a Boolean function (e.g. \cite{odonnell_book}).
\begin{definition}[\Fstab]
\label{def:four_stab}
A function $g(t): [-1,1]^n \rightarrow [-1, 1]$ is algebraically $k$-stable if for every $t \in [-1, 1]^n$ there exists a Boolean extension $h_t(z)$ at $t$ such that:
\begin{align}
\Four{h_t}(S) &\ge 0 \text{ for all } S \sse [n] \\
\widehat{h_t}(S) &= 0 \text{ for all } S \text{ such that } |S| > k,
\end{align}
where $\widehat{h_t}(S)$ is the Fourier coefficient of $h_t$ at $S \sse [n]$.
%\footnote{This notion can be relaxed by only requiring $\widehat{h_t}(S)$ is close to zero. We leave this relaxation to the full version of the paper.}
\end{definition}

In the parlance of Boolean function analysis, $g$ is
\emph{algebraically stable}
if for all $t$, the Fourier spectrum of $h_t$ is both nonnegative and \emph{low-degree}. %Functions with a low-degree spectrum are known to be efficiently learnable (see e.g. \cite{yishay}). Algorithms for learning such functions involve enumerating small sets $|S| \leq k$, and estimating each corresponding Fourier coefficient using random sampling and tail bounds. 
%
% In the literature of Boolean function analysis, this is almost the same as saying $h_t(z)$ has low-degree spectrum. It is well-known \cite{yishay} that low-degree functions are easy to learn with a sampling oracle.
% These algorithms essentially enumerate all small sets $S \sse [n]$ with $|S| \le k$ and estimate each $\hat{f}(S)$ using a Chernoff-type argument. Their goal is to bound the $\mathcal{L}^p$ distance between the output hypothesis and ground truth.
% Although our definition of \Fstab resembles this line of work, the techniques for proving that sampling works for mixture selection and signaling problems are quite different.
We prove the following  analogue of Theorem~\ref{thm:main}. %, the following theorem relates the stability `$k$' and the necessary sample size for linear program \eqref{lp:main}.

\begin{theorem}
  \label{thm:four_main}
Let $g:\gdom \rightarrow [-1, 1]$ be algebraically $k$-stable and \cLip in $L^\infty$,
let $A$ be an $\numrow \times \numcol$ matrix with entries in $[-1, 1]$,
let $\delta,\eps > 0$,
and let  $s > \frac{1}{\delta^2} \cdot \log \frac{2k}{\eps}$ be an integer. 

Fix a vector $x \in \Delta_\numcol$,
and let the random variable $\tilde{x}$
denote the empirical distribution of $s$ i.i.d.\ samples
from probability distribution $x$.
The following then holds:
\begin{equation} 
\Exp_{\tilde{x}} [g(A  \tilde{x})] \ge g(A x) - \eps - c \delta,
\end{equation}
\end{theorem}

% \begin{theorem}
% \label{thm:four_main}
% If $g$ is algebraically $\kStab$-stable and $c$-Lipschitz in $\mathcal{L}^{\infty}$, then with sample size $s = O(\frac{1}{\delta^2} \cdot \log \frac{k}{\eps})$
% \begin{equation} 
% \Exp_{\tilde{x}} [g(A \cdot \tilde{x})] \ge (1 - \eps)g(A \cdot x) - c \delta,
% \end{equation}
% where $\tilde{x}$ is the empirical distribution of $s$ i.i.d. samples from $x \in \Delta_m$.
% \end{theorem}

\begin{proof}[Proof sketch]
Let $t = A  x$.
Consider the random variable $\tilde{t} = A  \tilde{x}$ where $\tilde{x}$ is the (random) empirical distribution. Let $E_i$ denote the event that $|\tilde{t}_i - t_i| \le \delta$ for some $\delta \in (0, 1)$, i.e., that coordinate $t_i$ is approximately preserved. 
Define a Boolean function $\pD: \HC \rightarrow [0, 1]$,
\begin{equation}
\pD(z) := \Pr_{\tilde{x}}[(\cap_{i \in \III^+(z)} E_i)  \cap (\cap_{i \in \III^{-}(z)} \bar{E_i}) ].
\end{equation}
In particular, $\pD(z)$ is the probability that \emph{only} coordinates in $\III^+(z)$ are approximately preserved. It is easy to see that $\sum_{z \in \HC} \pD(z) = 1$ and $\Pr[E_i] = \Pr[|\tilde{t_i} - t| \le \delta] = \sum_{z: z_i = 1} \pD(z)$.

Next, we state some standard equalities for Boolean functions.
Let $\Bone$ denote the all-one $n$-dimensional vector $(1, 1, \ldots, 1)$ and $h_t$ be a Boolean extension of $g_t$
as stated in Definition~\ref{def:four_stab}.
\begin{equation}
\label{eq:g_exp}
\sum_{z \in \HC}  h_t(z) \cdot \pD(z) = 2^n \sum_{S \sse [n]} \Four{\pD}(S) \cdot \Four{h_t}(S) 
\end{equation}
\begin{equation}
\sum_{S \sse [n]} \Four{h_t}(S) = h_t(\Bone)
\end{equation}
By Definition~\ref{def:bool_ext}, $h_t(\Bone) = g(t)$.
Since $g(t)$ is $c$-Lipschitz in $L^{\infty}$, we know:
\begin{equation}
\label{eq:four_lip}
\Exp_{\tilde{x}}[g(A \tilde{x})] = \Exp_{\tilde{t}} [g(\tilde{t})] \ge \sum_{z \in \HC}  h_t(z) \cdot \pD(z) - c \delta,
\end{equation}
so it suffices to prove 
$\sum_{S \sse [n]} \Four{h_t}(S) - 2^n \sum_{S \sse [n]}  \Four{\pD}(S) \cdot \Four{h_t}(S) \le \eps$, which can be rewritten as:
\begin{equation}
\label{eq:four_error}
  \sum_{S:|S| \le k} (1 - 2^n\Four{\pD}(S)) \Four{h_t}(S) + \sum_{S:|S|>k} (1 - 2^n\Four{\pD}(S)) \Four{h_t}(S)
\end{equation}
As the latter term in \eqref{eq:four_error} is zero by definition \ref{def:four_stab}, we only need to upper bound $\sum_{S:|S| \le k} (1 - 2^n\Four{\pD}(S)) \Four{h_t}(S)$.
Using a simple union-bound argument, one can prove that with sample size $s \ge \frac{1}{\delta^2} \log \frac{2k}{\eps}$ , 
$\Four{\pD}(S) \ge \frac{1 - \eps}{2^n}$ for all $S \sse [n]$ with $|S| \le k$. As $\Four{h_t}(S) \ge 0$ for $S$ with $|S| \le \kStab$,
we have $\sum_{S:|S| \le k} (1 - 2^n\Four{\pD}(S)) \Four{h_t}(S) \le \eps \sum_{S \sse [n]} \Four{h_t}(S) \le \eps$. Combined with equation $\eqref{eq:four_lip}$
we have $\Exp[g(A\tilde{x})] \ge g(t) - \eps - c \delta$.
\end{proof}

\subsection{\Fstab for our Applications}
We observe that algebraic stability holds for the objective functions in some our applications, and leave open whether it holds for all. Specifically,  the relevant functions in lottery design, revenue maximizing signaling, and one of our two voting problems, are algebraically $O(1)$-stable, and therefore an additive PTAS for each of those applications follows from Theorem~\ref{thm:four_main} just as it did from Theorem~\ref{thm:main}. We list the Boolean extension associated with each of these applications. 
%\sdcomment{Is low-degree spectrum obvious for each of these Boolean extensions? I haven't thought about it. As long as someone familiar with Boolean function analysis would think it is obvious, I think we are OK and don't need to add detail}
 %Here we show how algebraic stability can be used to solve the following problems: lottery design, persuasion in voting and the revenue maximization problem. [this sentence needs to be rephrased, maybe shaddin can help] The stability of General-sum games can be seen by using union bound, so we do not apply Boolean function analysis there.

\begin{lemma}[Lottery design]
For a fixed price $p$, a Boolean extension of the objective function $\gLot{\Weight,p}(t) := p \cdot \sum_{i = 1}^{n} (\Weight_i \cdot I[t_i \geq p])$ at $t$ is $h^{\text{(lottery)}}_t(z) = p \cdot \sum_{i \in [n]} w_i  \frac{z_i + 1}{2} I[t_i \ge p]$. We enumerate over all prices $p$, up to a suitable discretization, in the algorithmic solution.
\end{lemma}

% Here we prove a more general statement that: if $h_1, \cdots, h_l$ are \fstab, then their point-wise maximum $h := \max \{h_1, \cdots, h_l \}$ is also \fstab.

% \begin{equation}
%  \Four{h}(S) = \frac{\sum_{z \in \HC} h(z) \chi^S(z)}{2^n} = \frac{\sum_{z \in \HC} \max\{h_1(z),\cdots,h_l(z)\} \chi^S(z)}{2^n}
% \end{equation}

% Let $j(z) := \argmax_{j \in [l]} \{h_j(z)\}$. Then
% \begin{equation}
%   \Four{h}(S) = \frac{\sum_{z \in \HC} h_{j(z)}(z) \chi^S(z)}{2^n}
% \end{equation}

% As each $h_j$ is \fstab, we can write $h_j(z) = \sum_{S \sse [n]} \Four{h_j}(S) \chi^S(z)$.
% \begin{align}
%   \Four{h}(S) = \frac{\sum_{z \in \HC} (\sum_{T \sse [n]} \Four{h_{j(z)}}(T) \chi^T(z)) \chi^S(z)}{2^n} \\
% =\frac{1}{2^n}\sum_{T \sse [n]} \sum_{z \in \HC} \Four{h_{j(z)}}(T) \chi^T(z)\chi^S(z) = \frac{1}{2^n} \sum_{T \sse [n]} \Four{h_{j(z)}}(T)
% \end{align}
% Given that $max$ is a convex function, this follows from the above lemma for voting.

\begin{lemma}[Persuasion in voting]
A Boolean extension of $\gVotesum$ at $t$ is $h^{\text{(vote-sum)}}_t(z) = \sum_{i \in [n]} \frac{I[t_i \ge 0]}{n} \frac{z_i + 1}{2}$.
\end{lemma}

\begin{lemma}[Revenue maximization]
A Boolean extension of the function $\maxtwo:[0,1]^n \to [0,1]$ at $t$ is $h^{\text{(max2)}}_t(z) = t_j \frac{z_i + 1}{2} \frac{z_j + 1}{2}$, where $i$ and $j$ denotes the indices of the largest and second-largest entry of $t$, respectively.
\end{lemma}

%%% Local Variables: 
%%% mode: latex
%%% TeX-master: "main"
%%% End: 